\documentclass[]{svjour3}                     
\smartqed  
\usepackage{graphicx}
\usepackage{mathptmx}      
\journalname{Higher-Order Symb Comput}
\catcode`\@=11
\def \myverbatimindent {3em}
\let \saveverbatime \@xverbatim
\def \@xverbatim {\leftskip = \myverbatimindent\relax\saveverbatime}
\catcode`\@=12

\usepackage{mathpartir}

\usepackage{extpfeil}

\usepackage[english]{babel}

\usepackage{multirow}
\usepackage{booktabs} 

\usepackage[unicode=true, pdfborder={0 0 0}]{hyperref}

\makeatletter
\newcommand{\prw@zbreak}{\nobreak\hskip\z@skip}
\newcommand{\BreakableHyphen}{\leavevmode%
  \prw@zbreak-\discretionary{}{}{}\prw@zbreak}
\DeclareRobustCommand{\hyp}{%
  \ifmmode-\else\BreakableHyphen\fi}
\makeatother

\newcommand{\ssi}{\ensuremath{\quad\text{iff}\quad}}

\def\vcentcolon{\mathrel{\mathop\ordinarycolon}}
\newcommand{\coloneqq}{\vcentcolon\mkern-1.2mu=}
\newcommand{\Coloneqq}{\vcentcolon\mkern-.9mu\vcentcolon\mkern-1.2mu=}

\newcommand{\ajout}[3]{\ensuremath{#1 + \{#2\mapsto#3\} }}
\newcommand{\subst}[3]{\ajout{#1}{#2}{#3}}

\newcommand{\lift}[1]{\ensuremath{\mathop{(#1)}\nolimits_{\ast}}}

\newcommand{\dom}{\mathop{\mathrm{dom}}\nolimits}
\newcommand{\gc}[2]{#2 \setminus #1}
\newcommand{\Image}{\mathop{\mathrm{Im}}\nolimits}
\newcommand{\Loc}{\mathop{\mathrm{Loc}}\nolimits}
\newcommand{\Env}{\mathop{\mathrm{Env}}\nolimits}
\newcommand{\close}[1]{#1_{*}}

\newcommand{\concat}{\cdot}

\newcommand{\range}[3]{\ensuremath{#1_{#2},\dotsc,#1_{#3}}}
\newcommand{\drange}[4]{\ensuremath{(#1_{#3},#2_{#3})\concat\dotsc\concat(#1_{#4},#2_{#4})}}

\newcommand{\unit}{\ensuremath{\mathbf{1}}}
\newcommand{\true}{\ensuremath{\mathop{\mathbf{true}}\nolimits}}
\newcommand{\false}{\ensuremath{\mathop{\mathbf{false}}\nolimits}}
\newcommand{\ite}[3]{\ensuremath{\mathop{\mathbf{if}}\ #1\
 \mathop{\mathbf{then}}\ #2\ \mathop{\mathbf{else}}\ #3}}
\newcommand{\letrec}[3]{\ensuremath{\mathop{\mathbf{letrec}}\ #1=#2\
  \mathop{\mathbf{in}}\ #3}}
\newcommand{\expr}{\ensuremath{\mathop{expr}\nolimits\/}}
\newcommand{\F}{\ensuremath{\mathcal{F}}}
\newcommand{\fun}[3]{\ensuremath{\left[\lambda#1.#2,#3\right]}}
\newcommand{\rhot}{\ensuremath{{\rho_{T}}}}
\newcommand{\env}[2]{\ensuremath{#1|#2}}

\newextarrow{\intArrow}{55{40}0}{\equiv\equiv\Rrightarrow}
\newextarrow{\optArrow}{55{40}0}{\Relbar\Relbar\Rightarrow}
\newextarrow{\naiveArrow}{55{40}0}{\relbar\relbar\rightarrow}

\newcommand{\reductionF}[6][\env{\rhot}{\rho}]{%
  \ensuremath{\left< #2,#3\right> \optArrow[#6]{#1} \left< #4, #5\right>}}
\newcommand{\reduction}[5][\env{\rhot}{\rho}]{\reductionF[#1]{#2}{#3}{#4}{#5}{\F}}

\newcommand{\reduclift}[5][\env{\rhot}{\rho}]{%
  \reductionF[#1]{\lift{#2}}{#3}{#4}{#5}{\lift{\F}}}

\newcommand{\reductionNF}[6][\rho]{%
  \ensuremath{\left< #2,#3\right> \naiveArrow[#6]{#1} \left< #4, #5\right>}}
\newcommand{\reductionN}[5][\rho]{\reductionNF[#1]{#2}{#3}{#4}{#5}{\F}}

\begin{document}

\title{Continuation-Passing C}
\subtitle{Compiling threads to events through continuations}

\author{Gabriel Kerneis       \and
        Juliusz Chroboczek
}

\institute{   G. Kerneis \at
              Laboratoire PPS, Universit\'e Paris Diderot,
              Case 7014,
              75205 Paris Cedex 13, France\\
              \email{kerneis@pps.univ-paris-diderot.fr}
           \and
              J. Chroboczek \at
              Laboratoire PPS, Universit\'e Paris Diderot,
              Case 7014,
              75205 Paris Cedex 13, France
}

\date{Received: date / Accepted: date}

\maketitle

\begin{abstract}

In this paper, we introduce Continuation Passing C (CPC), a programming
language for concurrent systems in which native and cooperative threads
are unified and presented to the programmer as a single abstraction.
The CPC compiler uses a compilation technique, based on the CPS
transform, that yields efficient code and an extremely lightweight
representation for contexts.
We provide a proof of the correctness of our compilation
scheme.  We show in particular that lambda-lifting, a common compilation
technique for functional languages, is also correct in an imperative
language like C, under some conditions enforced by the CPC compiler.
The current CPC compiler is mature enough to write substantial programs
such as Hekate, a highly concurrent BitTorrent seeder.  Our benchmark
results show that CPC is as efficient, while using significantly less
space, as the most efficient thread libraries available.

 \end{abstract}

\tableofcontents

\section{Introduction}

Most computer programs are {\em concurrent\/} programs, which need to
perform multiple tasks at a given time.  For example, a network server
needs to serve multiple clients at a time; a program with a graphical
user interface needs to expect keystrokes and mouse clicks at multiple
places; and a network program with a graphical interface (e.g.\ a web
browser) needs to do both.

The dominant abstraction for concurrency is provided by {\em threads},
or {\em lightweight processes}.  When programming with threads,
a process is structured as a dynamically changing number of
independently executing threads that share a single heap (and possibly
other data, such as global variables).  The alternative to threads is
{\em event-loop\/} or {\em event-driven\/} programming.  An event-driven
program interacts with its environment by reacting to a set of stimuli
called {\em events}.  At any given point in time, to every event is
associated a piece of code known as the {\em handler\/} for this event.
A global scheduler, known as the {\em event loop}, repeatedly waits for
an event to occur and invokes the associated handler.

Unlike threads, event handlers do not have an associated stack;
event-driven programs are therefore more lightweight and often faster
than their threaded counterparts.  However, because it splits the flow
of control into multiple tiny event-handlers, event-driven programming
is difficult and error-prone.  Additionally, event-driven programming
alone is often not powerful enough, in particular when accessing
blocking APIs or using multiple processor cores; it is then necessary to
write \emph{hybrid} code, that uses both threads and event handlers,
which is even more difficult.

Since event-driven programming is more difficult but more efficient than
threaded programming, it is natural to want to at least partially
automate it.

\paragraph{Continuation Passing C}

Continuation Passing C (CPC) is an extension of the C programming
language with concurrency primitives which is implemented by a series of
source-to-source transformations, including a phase of lambda-lifting
(Section~\ref{sec:lifting-intro}) and a transformation into Continuation-Passing Style
(Section~\ref{sec:cps-conversion}).  After this series of transformations, the CPC
translator yields a program in hybrid style, most of it event-driven, but
using the native threads of the underlying operating system wherever
this has been specified by the CPC programmer, for example in order to
call blocking APIs or distribute CPU-bound code over multiple cores or
processors.

Although it is implemented as a series of well-understood
source-to-source translations, CPC yields code that is just as fast as
the fastest thread libraries available to us; we further
believe that the code it generates is similar to carefully tuned code
written by a human programmer with a lot of experience with event-driven
code.  And although it is an extension of the C programming language,
a language notoriously hostile to formal proofs, the correctness of the
transformations used by CPC have been proved; while we
perform these proofs in a simplified dialect, there is nothing in
principle that prevents a complete proof of the correctness of CPC.  We
believe that this property makes CPC unique among the systems for
efficient implementation of concurrency in a widely used imperative
programming language.

In this paper, we give an introduction to CPC programming, we describe
informally the transformations performed by CPC, and give an outline of
the proof of correctness; the complete proof can be found in a companion
technical report \cite{kerneis12}.

\subsection{Outline}

The rest of this paper is organised as follows.  In
Section~\ref{sec:related-work}, we present previous work related to
threads and event-driven programming, and to the compilation techniques
used by the CPC translator.  In Section~\ref{sec:cpc-language}, we give
an overview of the CPC language through excerpts of code from Hekate,
the most significant program written in CPC so far.  In
Section~\ref{sec:experimental-results}, we show that CPC threads are as
fast as the most efficient thread libraries available to us.  In
Section~\ref{sec:compilation-technique}, we detail the several passes of
the CPC compilation technique to translate threads into events, and in
Section~\ref{sec:lifting} we prove the correctness of one of these
passes, lambda-lifting, in the context of an imperative, call-by-value
language restricted to functions called in tail position.  We conclude
in Section~\ref{sec:conclusion}.

\section{Related work}
\label{sec:related-work}

\subsection{Continuations and concurrency}

Continuations offer a natural framework to implement concurrency systems
in functional languages: they capture ``the rest of the computation''
much like the state of an imperative program is captured by the call
stack.  Thread-like primitives may be built with either first-class
continuations, or encapsulated within a continuation monad.

The former approach is best illustrated by \emph{Concurrent ML} constructs
\cite{reppy}, implemented on top of SML/NJ's first-class continuations,
or by the way coroutines and engines are typically implemented in Scheme using
the \texttt{call/cc} operator \cite{haynes,dybvig} (previously known as
\texttt{catch} \cite{wand}).  \emph{Stackless Python} \cite{tismer} uses
first-class continuations to implement generators, which are in turn used to
implement concurrency primitives.  Scala also uses first-class continuations,
through the \texttt{shift} and \texttt{reset} operators, to implement
concurrency primitives and asynchronous IO \cite{rompf}.

Explicit translation into continuation-passing style, often encapsulated
within a monad, is used in languages lacking first-class continuations.
In Haskell, the original idea of a concurrency monad is due to Scholz
\cite{scholz}, and extended by Claessen \cite{claessen} to a monad
transformer yielding a concurrent version of existing monads.  Li et.\ al
\cite{li} use a continuation monad to lazily translate the thread
abstraction exposed to the programmer into events scheduled by an event
loop.  In Objective Caml, Vouillon's \emph{lwt} \cite{vouillon} provides a
lightweight alternative to native threads, with the ability to execute
blocking tasks in a separate thread pool.

\subsection{Transformation techniques}

The three main techniques used by CPC --- CPS conversion, lambda-lifting and
splitting --- are fairly standard techniques for compiling functional
languages, or at least languages providing inner functions.

The conversion into continuation-passing style, or CPS conversion, has been
discovered and studied many times in different contexts
\cite{plotkin,strachey,reynolds}.  It is used for instance to compile Scheme
(Rabbit \cite{steele}) and ML (SML/NJ \cite{appel}), both of them exposing
continuations to the programmer.

Lambda-lifting, also called closure-conversion, is a standard technique to
remove free variables.  It is introduced by Johnsson \cite{johnsson} and made
more efficient by Danvy and Schultz \cite{danvy}.  Fischbach and Hannan prove
its correctness for the Scheme language\cite{fischbach}.  Although environments
are a more common way to handle free variables, some implementations use
lambda-lifting; for instance, the \emph{Twobit} Scheme-to-C compiler
\cite{clinger}.

We call \emph{splitting} the conversion of a complex flow of control into
mutually recursive function calls.  Van Wijngaarden is the first to describe
such a transformation, in a source-to-source translation for Algol 60
\cite{wijngaarden}.  The idea is then used by Landin to formalise a translation
between Algol and the lambda-calculus \cite{landin}, and by Steele and Sussman
to express gotos in applicative languages such as LISP or Scheme
\cite{steele}.  Thielecke adapts van Wijngaarden's transformation to the C
language, albeit in a restrictive way \cite{thielecke}.

We are aware of two implementations of splitting: Weave and TameJS.
\emph{Weave} is an unpublished tool used at IBM around year 2000 to write
firmware and drivers for SSA-SCSI RAID storage adapters \cite{weave}.  It
translates annotated Woven-C code, written in threaded style, into C code
hooked into the underlying event-driven kernel.
\emph{TameJS}\footnote{\url{http://tamejs.org/}}, by the authors of the C++
\emph{Tame} library \cite{krohn}, is a similar tool for Javascript where
event-driven programming is made mandatory by the lack of concurrency
primitives.\footnote{Note that, contrary to TameJS, the original Tame
    implementation in C++ does not use splitting but a state machine using
    switches.}

\subsection{Threads and events}
There has been a lot of research to provide efficient  threads, as well
as to make event-driven programming easier \cite{behren,welsh,pai}.
We focus here on results involving a transformation between threads and
events, or building bridges between them.

Adya et\ al. \cite{adya} present adaptors between event-driven and
threaded code to write programs mixing both styles. They first introduce
the notion of \emph{stack ripping} to describe the manual stack
management required by event-driven style.

Duff introduces a technique, known as \emph{Duff's device} \cite{duff}, to
express general loop unrolling directly in C, using the \verb+switch+ statement.
Since then, this technique has been employed multiple times to express state
machines and event-driven programs in a threaded style.\footnote{This abuse was
    already envisioned by Duff in 1983: ``I have another revolting way to use
    switches to implement interrupt driven state machines but it's too horrid to
    go into.'' \cite{duff}} For instance, it is used by Tatham to implement
coroutines in C \cite{tatham}.  Other C libraries later expanded this idea, such
as \emph{protothreads} \cite{protothreads} and \emph{FairThreads}' automata
\cite{boussinot}.  These libraries help keep a clearer flow of control but they
provide no automatic handling of local variables: the programmer is expected to
save them manually in his own data structures, just like in event-driven style.

\emph{Tame} \cite{krohn} is a C++ language extension and library which exposes events
to the programmer while introducing state machines to avoid the stack ripping
issue and retain a thread-like feeling.  The programmer needs to annotate
local variables that must be saved across context switches.

\emph{TaskJava} \cite{fischer} implements the same idea than Tame, in Java, but
preserves local variables automatically, storing them in a state record.
\emph{Kilim} \cite{srinivasan} is a message-passing framework for Java providing
actor-based, lightweight threads.  It is also implemented by a partial CPS
conversion performed on annotated functions, but contrary to TaskJava, it works
at the JVM bytecode level.

\emph{AC} is a set of language constructs for composable asynchronous IO in C
and C++ \cite{harris}.  Harris et\ al. introduce \texttt{do..finish} and
\texttt{async} operators to write asynchronous requests in a synchronous style,
and give an operational semantics.  The language constructs are somewhat similar
to those of Tame but the implementation is widely different, using LLVM code
blocks or macros based on GCC nested functions rather than source-source
transformations.

Haller and Odersky \cite{haller} advocate unification of thread-based and
event-based models through actors, with the \texttt{react} and \texttt{receive}
operators provided by the \emph{Scala Actors} library; suspended actors are
represented by continuations.

\section{The CPC language} \label{sec:cpc-language}

Together with two undergraduate students, we taught ourselves how to program in
CPC by writing \emph{Hekate}, a \emph{BitTorrent} \emph{seeder}, a massively
concurrent network server designed to efficiently handle tens of thousands of
simultaneously connected peers \cite{places2011,attar-canal}.  In this section,
we give an overview of the CPC language through several programming idioms
that we discovered while writing Hekate.

\subsection{Cooperative CPC threads}
\label{sec:cpc-threads}

The extremely lightweight, cooperative threads of CPC lead to a ``threads
are everywhere'' feeling that encourages a somewhat unusual programming
style.

\paragraph{Lightweight threads}
Contrary to the common model of using one thread per client, Hekate
spawns at least three threads for every connecting peer: a reader, a
writer, and a timeout thread.  Spawning several CPC threads per client
is not an issue, especially when only a few of them are active at any
time, because idle CPC threads carry virtually no overhead.

The first thread reads incoming requests and manages the state of the
client.  The BitTorrent protocol defines two states for interested peers:
\emph{``unchoked,''} i.e.\ currently served, and \emph{``choked.''}  Hekate
maintains 90\,\% of its peers in \emph{choked} state, and \emph{unchokes} them
in a round-robin fashion.

The second thread is in charge of actually sending the chunks of data
requested by the peer.  It usually sleeps on a condition variable,
and is woken up by the first thread when needed.  Because these
threads are scheduled cooperatively, the list of pending chunks is
manipulated by the two threads without need for a lock.

Every read on a network interface is guarded by a timeout, and a peer that
has not been involved in any activity for a period of time is disconnected.
Earlier versions of Hekate which did not include this protection would end
up clogged by idle peers, which prevented new peers from connecting.

In order to simplify the protocol-related code, timeouts are implemented in
the buffered read function, which spawns a new timeout thread on each
invocation.  This temporary third thread sleeps for the duration of the timeout,
and aborts the I/O if it is still pending.  Because most timeouts do not
expire, this solution relies on the efficiency of spawning and
context-switching short-lived CPC threads (see
Section~\ref{sec:experimental-results}).

\paragraph{Native and cps functions}
CPC threads might execute two kinds of code: \emph{native} functions
and \emph{cps} functions (annotated with the \texttt{cps} keyword).
Intuitively, cps functions are interruptible and native functions are not:
it is possible to interrupt the flow of a block of cps code in order to
pass control to another piece of code, to wait for an event to happen or
to switch to another scheduler (Section~\ref{sec:detached-threads}).
Note that the \texttt{cps} keyword does not mean that the function is
written in continuation-passing style, but rather that it is to be
CPS-converted by the CPC translator.
Native code, on the other hand, is ``atomic'': if a sequence of native
code is executed in cooperative mode, it must be completed before anything
else is allowed to run in the same scheduler.
From a more technical point of view, cps functions are compiled by
performing a transformation to Continuation Passing Style (CPS), while
native functions execute on the native stack.

There is a global constraint on the call graph of a CPC program: a cps
function may only be called by a cps function; equivalently, a native
function can only call native functions --- but a cps function may call
a native function.  This means that at any point in time, the dynamic
chain consists of a ``cps stack'' of cooperating functions followed by a
``native stack'' of regular C functions.  Since context switches are
forbidden in native functions, only the cps stack needs to be saved and
restored when a thread cooperates.

\paragraph{CPC primitives}
CPC provides a set of primitive cps functions, which allow the
programmer to schedule threads and wait for some events. These primitive
functions could not have been defined in user code: they must have
access to the internals of the scheduler to operate. Since they are cps
functions, they can only be called by another cps function.

Figure~\ref{fig:listening} shows an example of a cps function:
\texttt{listening} calls the primitive \texttt{cpc\_io\_wait} to wait
for the file descriptor \texttt{socket\_fd} to be ready, before
accepting incoming connections with the native function \texttt{accept}
and spawning a new thread for each of them.  Note the use of the
\texttt{cpc\_spawn} keyword to create a new thread executing the cps
function \texttt{client}.

\begin{figure}
\begin{center}
\begin{minipage}{0.7\linewidth}
\small
\begin{verbatim}
cps void
listening(hashtable * table) {
    /* ... */
    while(1) {
        cpc_io_wait(socket_fd, CPC_IO_IN);
        client_fd = accept(socket_fd, ...);
        cpc_spawn client(table, client_fd);
    }
}
\end{verbatim}
\end{minipage}
\end{center}
\caption{Accepting connections and spawning threads}
\label{fig:listening}
\end{figure}

The CPC language provides five cps primitives to suspend and synchronise
threads on some events.  The simplest one is \texttt{cpc\_yield}, which
yields control to the next thread to be executed.  The primitives
\texttt{cpc\_io\_wait} and \texttt{cpc\_sleep} suspend the current
thread until a given file descriptor has data available or some time has
elapsed, respectively.  A thread can wait on some condition variable
\cite{hoare} with \texttt{cpc\_wait} ; threads suspended on a condition
variable are woken with the (non-cps) functions \texttt{cpc\_signal}
and \texttt{cpc\_signal\_all}.  To allow early interruptions,
\texttt{cpc\_io\_wait} and \texttt{cpc\_sleep} also accept an optional
condition variable which, if signaled, will wake up the waiting thread.

The fifth cps primitive, \texttt{cpc\_link}, is used to control how
threads are scheduled.  We give more details about it in
Section~\ref{sec:detached-threads}.

We found that these five primitives are enough to build more complex
synchronisation constructs and cps functions, such as barriers or
retriggerable timeouts.  Some of these generally useful functions,
written in CPC and built above the CPC primitives, are distributed with
CPC and form the CPC standard library.

\subsection{Comparison with event-driven programming}
\label{sec:comparison}

\paragraph{Code readability}

Hekate's code is much more readable than its event-driven equivalents.
Consider for instance the BitTorrent handshake, a message exchange
occurring just after a connection is established.  In
\emph{Transmission}\footnote{\url{http://www.transmissionbt.com/}},
a popular and efficient BitTorrent client written in (mostly)
event-driven style, the handshake is a complex piece of code, spanning
over a thousand lines in a dedicated file.  By contrast, Hekate's
handshake is a single function of less than fifty lines including error
handling.

While some of Transmission's complexity is explained by its support for
encrypted connexions, Transmission's code is intrinsically much more
messy due to the use of callbacks and a state machine to keep track of
the progress of the handshake.  This results in an obfuscated flow of
control, scattered through a dozen of functions (excluding
encryption-related functions), typical of event-driven code.

\paragraph{Expressivity}

Surprisingly enough, CPC threads turn out to naturally express some idioms that
are more commonly associated with event-driven style.

A case in point: buffer allocation for reading data from the network.
When a native thread performs a blocking read, it needs to allocate the
buffer before the \texttt{read} system call; when many threads are
blocked waiting for a read, these buffers add up to a significant amount
of storage.  In an event-driven program, it is possible to delay
allocating the buffer until after an event indicating that data is
available has been received.

The same technique is not only possible, but actually natural in CPC:
buffers in Hekate are only allocated after \texttt{cpc\_io\_wait} has
successfully returned.  This provides the reduced storage requirements
of an event-driven program while retaining the linear flow of control of
threads.

\subsection{Detached threads}
\label{sec:detached-threads}

While cooperative, deterministically scheduled threads are less error-prone
and easier to reason about than preemptive threads, there are circumstances
in which native operating system threads are necessary.  In traditional
systems, this implies either converting the whole program to use native
threads, or manually managing both kinds of threads.

A CPC thread can switch from cooperative to preemptive mode at
any time by using the \texttt{cpc\_link} primitive (inspired by
FairThreads' \texttt{ft\_thread\_link} \cite{boussinot}).  A cooperative
thread is said to be \emph{attached} to the default scheduler, while
a preemptive one is \emph{detached}.

The \texttt{cpc\_link} primitive takes a single argument, a scheduler,
either the default event loop (for cooperative scheduling) or a thread pool
(for preemptive scheduling).  It returns the previous scheduler, which
makes it possible to eventually restore the thread to its original state.
Syntactic sugar is provided to execute a block of code in attached or
detached mode (\texttt{cpc\_attached}, \texttt{cpc\_detached}).

Hekate is written in mostly non-blocking cooperative style; hence, Hekate's
threads remain attached most of the time.  There are a few situations,
however, where the ability to detach a thread is needed.

\paragraph{Blocking OS interfaces}

Some operating system interfaces, like the \texttt{getaddrinfo} DNS
resolver interface, may block for a long time (up to several seconds).
Although there exist several libraries which implement equivalent
functionality in a non-blocking manner, in CPC we simply enclose the call
to the blocking interface in a \texttt{cpc\_detached} block (see
Figure~\ref{fig:cpc-detached}a).

Figure~\ref{fig:cpc-detached}b shows how \texttt{cpc\_detached} is
expanded by the translator into two calls to \texttt{cpc\_link}.  Note that
CPC takes care to attach the thread before returning to the caller
function, even though the \texttt{return} statement is inside the
\texttt{cpc\_detached} block.

\begin{figure}[htb]
\small\centering
\begin{tabular}{ l | l }
&
\verb+cpc_scheduler *s =+\\
\verb+cpc_detached {+&
\verb+    cpc_link(cpc_default_threadpool);+\\
\verb+    rc = getaddrinfo(name, ...)+&
\verb+rc = getaddrinfo(name, ...)+\\
\verb+    return rc;+&
\verb+cpc_link(s);+\\
\verb+}+&
\verb+return rc;+\\
\multicolumn{1}{c}{\normalsize(a)}&
\multicolumn{1}{c}{\normalsize(b)}
\end{tabular}
\normalsize
\caption{Expansion of \texttt{cpc\_detached} in terms of \texttt{cpc\_link}}
\label{fig:cpc-detached}
\end{figure}

\paragraph{Blocking library interfaces}
Hekate uses the \textit{curl} library
\footnote{\url{http://curl.haxx.se/libcurl/}} to contact BitTorrent
\emph{trackers} over HTTP.  Curl offers both a simple, blocking
interface and a complex, asynchronous one.  We decided to use the one
interface that we actually understand, and therefore call the blocking
interface from a detached thread.

\paragraph{Parallelism}
Detached threads make it possible to run on multiple processors or
processor cores.  Hekate does not use this feature, but a CPU-bound program
would detach computationally intensive tasks and let the kernel schedule
them on several processing units.

\subsection{Hybrid programming}
\label{sec:hybrid-threads}

Most realistic event-driven programs are actually \emph{hybrid} programs
\cite{pai,welsh}: they consist of a large event loop, and a number of
threads (this is the case, by the way, of the \emph{Transmission}
BitTorrent client mentioned above).  Such blending of native threads
with event-driven code is made very easy by CPC, where switching from
one style to the other is a simple matter of using the \verb|cpc_link|
primitive.

This ability is used in Hekate for dealing with disk reads.  Reading from
disk might block if the data is not in cache; however, if the data is
already in cache, it would be wasteful to pay the cost of a detached
thread.  This is a significant concern for a BitTorrent seeder because
the protocol allows requesting chunks in random order, making kernel
readahead heuristics inefficient.

\begin{figure}
\begin{center}
\begin{minipage}{0.9\linewidth}
\small
\begin{verbatim}
    prefetch(source, length);                           /* (1) */
    cpc_yield();                                        /* (2) */
    if(!incore(source, length)) {                       /* (3) */
        cpc_yield();                                    /* (4) */
        if(!incore(source, length)) {                   /* (5) */
            cpc_detached {                              /* (6) */
                rc = cpc_write(fd, source, length);
            }
            goto done;
        }
    }
    rc = cpc_write(fd, source, length);                 /* (7) */
done:
    ...
\end{verbatim}
\end{minipage}
\end{center}
{\small\em The functions \verb|prefetch| and \verb|incore| are thin
wrappers around the {\tt posix\_madvise} and {\tt mincore} system
calls.}
\caption{An example of hybrid programming (non-blocking read)}
\label{fig:non-blocking-read}
\end{figure}

The actual code is shown in Figure~\ref{fig:non-blocking-read}: it sends
a chunk of data from a memory-mapped disk file over a network socket.
In this code, we first trigger an asynchronous read of the on-disk data
(1), and immediately yield to threads servicing other clients (2) in
order to give the kernel a chance to perform the read.  When we are
scheduled again, we check whether the read has completed (3); if it has,
we perform a non-blocking write (7); if it hasn't, we yield one more
time (4) and, if that fails again (5), delegate the work to a native
thread which can block (6).

Note that this code contains a race condition: the prefetched block of
data could have been swapped out before the call to {\tt cpc\_write},
which would stall Hekate until the write completes.  Moreover, this race
condition is even more likely to appear as load increases and on devices with
constrained resources.  To avoid this race condition and ensure non-blocking
disk reads, one could use asynchronous I/O.  However, while the Linux kernel
does provide a small set of asynchronous I/O system calls, we found them
scarcely documented and difficult to use: they work only on some file
systems, impose alignment restrictions on the length and address of buffers, and
disable the caching performed by the kernel.  We have therefore not experimented
with them.

Note further that the call to {\tt cpc\_write} in the {\tt
    cpc\_detached} block (6) could be replaced by a call to {\tt write}:
we are in a native thread here, so the non-blocking wrapper is not
needed.  However, the CPC primitives such as \texttt{cpc\_io\_wait} are
designed to act sensibly in both attached and detached mode; this
translates to more complex functions built upon them, and {\tt
    cpc\_write} simply behaves as {\tt write} when invoked in detached
mode.  For simplicity, we choose to use the CPC wrappers throughout our
code.

\section{Experimental results} \label{sec:experimental-results}

The CPC language provides no more than half a dozen very low-level
primitives.  The standard library and CPC programs are implemented in
terms of this small set of operations, and are therefore directly
dependent on their performance.  In Section~\ref{sec:speed-primitives},
we show the results of benchmarking individual CPC primitives against
other thread libraries.  In these benchmarks, CPC turns out to be
comparable to the fastest thread libraries available to us.

In the absence of reliable information on the relative dynamic frequency
of cps function calls and thread synchronisation primitives in CPC
programs, it is not clear what the performance of individual primitives
implies about the performance of a complete program.  In
Section~\ref{sec:macrobenchmarks}, we present the results of
benchmarking a set of naive web servers written in CPC and in a number
of thread libraries.

Finally, in Section~\ref{sec:bench-hekate}, we present a number of
insights that we gained by working with \emph{Hekate}, our BitTorrent
seeder written in CPC.

We compare CPC with the following thread libraries: \emph{nptl}, the
native thread library in GNU libc 2.13 (or $\mu$Clibc 0.9.32 for our
tests on embedded hardware) \cite{nptl}; \emph{GNU Pth} version 2.0.7
\cite{engelschall}; \emph{State Threads (ST)} version 1.9
\cite{shekhtman}.  \emph{Nptl} is a kernel thread library, while
\emph{GNU Pth} and \emph{ST} are cooperative user-space thread
libraries.

\subsection{Speed of CPC primitives} \label{sec:speed-primitives}

\subsubsection{Space utilisation}

On a 64-bit machine with 4\,GB of physical memory and no swap space, CPC can
handle up to 50.1 million simultaneous threads, which implies an average
memory usage of roughly 82 bytes per continuation.  This figure compares
very favourably to both kernel and user-space thread libraries (see
Figure~\ref{fig:space-benchmarks}), which our tests have shown to be
limited on the same system to anywhere from 32\,000 to 934\,600 threads
in their default configuration, and to 961\,400 threads at most after
some tuning.

\subsubsection{Speed}

Figure~\ref{fig:micro-benchmarks} presents timings on three different
processors: the \emph{Core 2 Duo}, a superscalar out-of-order 64-bit
processor, the \emph{Pentium-M}, a slightly less advanced out-of-order
32-bit processor, and the \emph{MIPS 4Kc}, a simple in-order 32-bit RISC
chip with a short pipeline.

\renewcommand{\thefootnote}{\em\alph{footnote}}

\begin{figure}[htb]
\begin{center}
\begin{minipage}{25em}
\renewcommand{\footnoterule}{}
\begin{center}
    \begin{tabular}{lr@{ }l}
\toprule
Library & Number of threads\\
\midrule
nptl & 32\,330\\
Pth & 700\,000 &(est.)\,\footnotemark[1]\\
ST & 934\,600\\
ST (4\,kB stacks) & 961\,400\\
CPC & 50\,190\,000\\
\bottomrule
\end{tabular}
\end{center}
\end{minipage}
\end{center}
\caption{Number of threads possible with 4 GB of memory\label{fig:space-benchmarks}}
\end{figure}

\begin{figure}[htb]
\begin{center}
\begin{minipage}{40em}
\begin{center}
    \begin{tabular}{lrrrrrr@{}l}
        \toprule
        Library & \centering loop & \centering call\footnotemark[2] & \centering cps-call & \centering switch &
        \centering cond & \centering spawn
        \tabularnewline
        \midrule
        nptl (1 core)& 2 & 2 & &439 & 1\,318 & 4\,000\,000&\,\footnotemark[3]\\
        nptl (2 cores)& 2  & 2 & & 135 & 1\,791 & 11\,037\\
        Pth & 2 & 2 & & 2\,130 & 2\,602 & 6\,940\\
        ST & 2 & 2 & & 170 & 25 & 411\\
        CPC & 2 & 2 & 20 & 16 & 36 & 74\\
        \bottomrule
        \addlinespace
        \multicolumn{6}{c}{\em x86-64: Intel Core 2 Duo at 3.17 GHz, Linux 2.6.39}
        \tabularnewline
        \addlinespace[1.5em]
        \toprule
        Library & \centering loop & \centering call & \centering cps-call & \centering switch &
        \centering cond & \centering spawn
        \tabularnewline
        \midrule
        nptl & 2 & 4 & & 691 & 2\,426 & 24\,575\\
        ST & 2 & 4 &  & 1\,003 & 96 & 2\,293\\
        CPC & 2 & 4 & 54 & 46 & 91 & 205\\
        \bottomrule
        \addlinespace
        \multicolumn{6}{c}{\em x86-32: Intel Pentium M at 1.7 GHz, Linux 2.6.38}
        \tabularnewline
        \addlinespace[1.5em]
        \toprule
        Library & \centering loop & \centering call & \centering cps-call & \centering switch &
        \centering cond & \centering spawn
        \tabularnewline
        \midrule
        nptl & 58 & 63 & & 6\,310 & 63\,305 & 933\,689\\
        CPC & 58 & 63 & 2\,018 & 1\,482 & 3\,268 & 8\,519\\
        \bottomrule
        \addlinespace
        \multicolumn{6}{c}{\em MIPS-32: MIPS 4KEc at 184 MHz, Linux 2.6.37}
    \end{tabular}
\end{center}

All times are in nanoseconds per loop iteration, averaged over millions
of runs (smaller is better).\\
The columns are as follows:\\
\emph{loop:} time per loop iteration (reference);\\
\emph{call:} time per loop iteration calling a direct-style function;\\
\emph{cps-call:} time per loop iteration calling a CPS-converted function;\\
\emph{switch:} context switch;\\
\emph{cond:} context switch on a condition variable;\\
\emph{spawn:} thread creation.
\end{minipage}
\end{center}
\caption{Speed of thread primitives on various architectures\label{fig:micro-benchmarks}}
\end{figure}
\footnotetext[1]{\emph{Pth} never completed, the given value is an educated guess.}
\footnotetext[2]{On x86-64, performing a direct-style function call
  adds no measurable delay.}
\footnotetext[3]{The behaviour of the system call \texttt{sched\_yield}
  has changed in Linux 2.6.23, where it has the side-effect of reducing
  the priority of the calling thread.  Until Linux 2.6.38, setting the
  kernel variable \texttt{sched\_compat\_yield} restored the previous
  behaviour, which we have done in our x86-32 and MIPS-32 benchmarks.
  This knob having been removed in Linux 2.6.39, the \emph{nptl} results
  on x86-64 are highly suspect.}
\renewcommand{\thefootnote}{\arabic{footnote}}

\paragraph{Function calls}

Modern processors include specialised hardware for optimising function
calls; obviously, this hardware is not able to optimise CPS-converted
function calls, which are therefore dramatically slower than
direct-style (non-cps) calls.  In our very simple benchmark (two embedded loops
containing a function call), the \emph{Core 2 Duo} was able to
completely hide the cost of a trivial function call\footnote{Disassembly
  of the binary shows that the function call is present.}, while the
function call costs three processor cycles on the \emph{Pentium-M} and
just one on the MIPS processor; we do not claim to understand the
hardware-level magic that makes such results possible.

With CPS-converted calls, on the other hand, our benchmarking loop takes
on the order of 100 processor cycles per iteration (400 on MIPS).  This
apparently poor result is mitigated by the fact that after goto
elimination and lambda-lifting, the loop consists of four CPS function
calls; hence, on the \emph{Pentium-M}, a CPS function call is just
6 times slower than a direct-style call, a surprisingly good result
considering that our continuations are allocated on the heap
\cite{miller}.  On the MIPS chip the slowdown is closer to a factor of
100, more in line with what we expected.

\paragraph{Context switches and thread creation}

The ``real'' primitives are pleasantly fast in CPC.  Context switches
have similar cost to function calls, which is not surprising since they
compile to essentially the same code.  While the benchmarks make them
appear to be much faster than in any other implementation, including \emph{ST},
this result is skewed by a flaw in the \emph{ST} library, which goes through
the event loop on every \verb|st_yield| operation.  Context switches on
a condition variable are roughly two times slower than pure context
switches, which yields timings comparable to those of \emph{ST}.

Thread creation is similarly fast, since it consists in simply
allocating a new continuation and registering it with the event loop.
This yields results that are ten times faster than \emph{ST}, which must
allocate a new stack, and a hundred times faster than the kernel
implementation.

\subsection{Macro-benchmarks} \label{sec:macrobenchmarks}

As we have seen above, the speed of CPC primitives ranges from similar
to dramatically faster than that of their equivalents in traditional
implementations.  However, a CPC program incurs the overhead of the CPS
translation, which introduces CPS function calls which are much slower
than their direct-style equivalents.  As it is difficult to predict the
number of CPS-converted calls in a CPC program, it is not clear whether
this performance increase will be reflected in realistic programs.

We have written a set of naive web servers (less than 200 lines each)
that share the exact same structure: a single thread or process waits
for incoming connections, and spawns a new thread or process as soon as
a client is accepted (there is no support for HTTP/1.1 persistent
connections).  The servers were benchmarked by repeatedly downloading
a tiny file over a dedicated Ethernet with different numbers of
simultaneous clients and measuring the average response time.  Because
of the simple structure of the servers, this simple, repeatable
benchmark measures the underlying implementation of concurrency rather
than the performance tweaks of a production web server.

\begin{figure}[htb]
    \begin{center}
      \includegraphics[width=75mm]{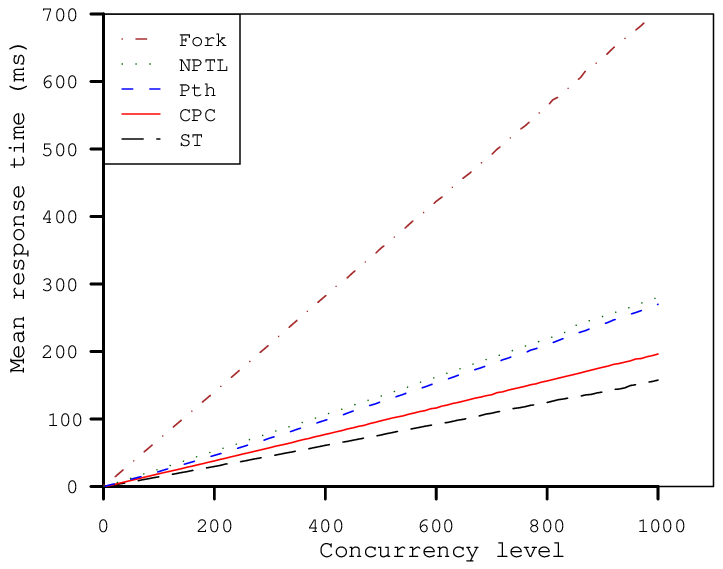}
  \end{center}
  \vspace*{-0.5cm}
  \caption{Web servers comparison}
  \label{fig:compare}
\end{figure}

Figure~\ref{fig:compare} presents the results of our experiment.  It
shows the average serving time per request against the number of
concurrent requests; a smaller slope indicates a faster server.  It
shows that on this particular benchmark all of our web servers show the
expected linear behaviour, and that CPC's performance is comparable to
that of \emph{ST}, and much better than that of the other implementations.
A more in-depth analysis of this benchmark is available in a technical
report \cite{kerneis}.

\subsection{Hekate}
\label{sec:bench-hekate}

Benchmarking a BitTorrent seeder is a difficult task because it relies
either on a real-world load, which is hard to control and only provides
seeder-side information, or on an artificial testbed, which might fail
to accurately reproduce real-world behaviour.  We have however been able
to collect enough empirical evidence to show that Hekate
(Section~\ref{sec:cpc-language}) is an efficient implementation of
BitTorrent, lightweight enough to be run on embedded hardware.

\paragraph{Real-world workload}
To benchmark the ability of Hekate to sustain a real-world load, we need
popular torrents with many requesting peers over a long period of time.
Updates for Blizzard's game \textit{World of Warcraft} (WoW), distributed
over BitTorrent, meet these conditions: each of the millions of WoW players
around the world runs a hidden BitTorrent client, and at any time many of
them are looking for the latest update.

We have run an instance of Hekate seeding WoW updates without
interruption for weeks.  We saw up to 1,000 connected peers (800 on
average) and a throughput of up to 10\,MB/s (around 5\,MB/s on average).
Hekate never used more than 10\,\% of the 3.16\,GHz dual core CPU of our
benchmarking machine, and was bottlenecked either by the available bandwidth
during peaks of requests (10\,MB/s), or by the mere fact that we did not gather
enough peers to saturate the link.

\paragraph{Stress-test on embedded hardware}

We have ported Hekate to \emph{OpenWrt}\footnote{\url{http://openwrt.org/}},
a Linux distribution for embedded devices.  Hekate runs flawlessly on
a cheap home router with a 266\,MHz \emph{MIPS 4Kc} CPU, 32\,MB of RAM
and a 100\,Mbps network card.  The torrent files were kept on a USB
flash drive.

Our stress-test consists in 1,000 simultaneous clients (implemented as
a single CPC program running on a fast computer directly connected to
the device over 100\,Mbps Ethernet) requesting random chunks of
a 1.2\,GB torrent.  In these conditions, Hekate was able to serve data
at a rate of 2.9\,MB/s.  The CPU load was pegged at 100\%, with most of
the CPU time spent servicing software interrupt requests, notably from
the USB subsystem (60\,\% \textit{sirq}, the usb-storage kernel module
using up to 25\,\% of CPU).  These data indicate that even on a slow
MIPS processor, Hekate's performance is fairly good, and that the
performance could be much improved by using a device on which mass
storage traffic doesn't go over the USB bus.

\section{The CPC compilation technique}
\label{sec:compilation-technique}

The current implementation of CPC is structured into three parts: the
CPC to C translator, implemented in Objective Caml \cite{ocaml} on top
of the CIL framework \cite{cil}, the runtime, implemented in C, and the
standard library, implemented in CPC.  The three parts are as
independent as possible, and interact only through a small set of
well-defined interfaces.  This makes it easier to experiment with
different approaches.

In this section, we present how the CPC translator turns a CPC program
in threaded style into an equivalent C program written in
continuation-passing style.  Therefore, we only need to focus on
the transformations applied to cps functions, ignoring completely the
notions of CPC thread or CPC primitive which are handled in the runtime
part of CPC.   We detail each step of the CPC compilation technique and
the difficulties specifically related to the C language.

\subsection{Translation passes}

The CPC translator is structured in a series of source-to-source
transformations which turn a threaded-style CPC program into an
equivalent event-driven C program.  This sequence of transformations
consists of the following passes:
\begin{description}
    \item[\em Boxing] a small number of variables needs to be encapsulated
        in environments to ensure the correctness of the later passes;
    \item[\em Splitting] the flow of control of each cps function is split
        into a set of mutually recursive, tail-called, inner functions;
    \item[\em Lambda-lifting] free local variables are copied from one inner
        function to the next, yielding closed inner functions;
    \item[\em CPS conversion] at this point, the program is in
        \emph{CPS-convertible form}, a form simple enough to
        perform a one-pass partial conversion into continuation-passing
        style; the resulting continuations are used at runtime by the
        CPC scheduler.
\end{description}
The converted program is then compiled by a standard C compiler and
linked to the CPC scheduler to produce the final executable.

All of these passes are well-known compilation techniques for functional
programming languages, but lambda-lifting and CPS conversion are not
correct in general for an imperative call-by-value language such as C.
The problem is that these transformations copy free variables: if the
original variable is modified, the copy becomes out-of-date and the
program yields a different result.

Copying is not the only way to handle free variables.  When applying
lambda-lifting to a call-by-value language with mutable variables, the common
solution is to box variables that are actually mutated, storing them in
environments.  However, this is less efficient: adding a layer of indirection
hinders cache efficiency and breaks a number of compiler optimisations.
Therefore, CPC strives to avoid boxing as much as possible.

One key point of the efficiency of CPC is that we need not box every
mutated variable for lambda-lifting and CPS conversion to be correct.
As we show in Section~\ref{sec:lifting}, even though C is an imperative
language, lambda-lifting is correct without boxing for most variables,
provided that the lifted functions are called in tail position.  As it turns
out, functions issued from the splitting pass are always called in tail
position, and it is therefore correct to perform lambda-lifting in CPC while
keeping most variables unboxed.

Only a small number of variables, whose addresses have been retained
with the ``address of'' operator (\verb+&+), need to be boxed: we call
them extruded variables.  Our experiments with Hekate show that 50\,\%
of local variables in cps functions need to be lifted.  Of that number,
only 10\,\% are extruded.  In other words, in the largest program
written in CPC so far, our translator manages to box only 5\,\% of the
local variables in cps functions.

Sections~\ref{sec:cps-conversion} to~\ref{sec:lifting-intro} present
each pass, explain how they interact and why they are correct.  Although CPS
conversion is the last pass performed by the CPC translator, we present it first
(Section~\ref{sec:cps-conversion}) because it helps understanding the purpose of
the other passes, which aim at translating the program into CPS-convertible
form; the other passes are then presented chronologically.  Also note that the
correctness of the lambda-lifting pass depends on a theorem that will be shown
in Section~\ref{sec:lifting}.

\subsection{CPS conversion}\label{sec:cps-conversion}

{\em Conversion into Continuation Passing Style\/}
\cite{strachey,plotkin}, or {\em CPS conversion\/} for short,
is a program transformation technique that makes the flow of control of
a program explicit and provides first-class abstractions for it.

Intuitively, the {\em continuation\/} of a fragment of code is an
abstraction of the action to perform after its execution.  For example,
consider the following computation:
\begin{verbatim}
f(g(5) + 4);
\end{verbatim}
The continuation of \texttt{g(5)} is \texttt{f($\cdot$ + 4)} because the
return value of \texttt{g} will be added to 4 and then passed to
\texttt{f}.

CPS conversion consists in replacing every function $f$ in a program
with a function $f^*$ taking an extra argument, its {\em continuation}.
Where $f$ would return with value $v$, $f^*$ invokes its continuation
with the argument $v$.
A CPS-converted function therefore never returns, but makes a call to
its continuation.  Since all of these calls are in tail position, a
converted program doesn't use the native call stack: the information
that would normally be in the call stack (the dynamic chain) is
encoded within the continuation.

This translation has three well-known interesting properties, on which
we rely in the implementation of CPC:
\begin{description}
    \item[\em CPS conversion need not be global]
The CPS conversion is not an ``all or nothing'' deal, in which the
complete program must be converted: there is nothing preventing a
converted function from calling a function that has not been converted.
On the other hand, a function that has not been converted cannot call a
function that has, as it does not have a handle to its own continuation.

It is therefore possible to perform CPS conversion on just a subset of
the functions constituting a program (in the case of CPC, such functions
are annotated with the \texttt{cps} keyword).  This allows cps code to
call code in direct style, for example system calls or standard library
functions.  Additionally, at least in the case of CPC, a cps function
call is much slower than a direct function call; being able to only
convert the functions that need the full flexibility of continuations
avoids this overhead as much as possible.

\item[\em Continuations are abstract data structures]

The classical definition of CPS conversion implements continuations as
functions.  However, continuations are abstract data structures and
functions are only one particular concrete representation of them.

The CPS conversion process performs only two operations on
continuations: calling a continuation, which we call \emph{invoke}, and
prepending a function application to the body of a continuation, which
we call \emph{push}.  This property is not really surprising: as
continuations are merely a representation for the dynamic chain, it is
natural that there should be a correspondence between the operations
available on a continuation and a stack.

Since C doesn't have full first-class functions (closures), CPC uses
this property to implement continuations as arrays of function pointers
and parameters.

\item[\em Continuation transformers are linear]
CPS conversion introduces linear (``one-shot'') continuations
\cite{berdine,bruggeman}: when a CPS-converted function receives a
continuation, it will use it exactly once, and never duplicate or
discard it.

This property is essential for memory management in CPC: as CPC uses
the C allocator (\verb|malloc| and \verb|free|) rather than a garbage
collector for managing continuations, it allows reliably reclaiming
continuations without the need for costly devices such as reference
counting.
\end{description}

CPS conversion is not defined for every C function; instead, we restrict
ourselves to a subset of functions, which we call the {\em
    CPS-convertible\/} subset.  As we shall see in
Section~\ref{sec:splitting}, every C function can be converted to an
equivalent function in CPS-convertible form.

The CPS-convertible form restricts the calls to cps functions to
make it straightforward to capture their continuation.  In
CPS-convertible form, a call to a cps function \texttt{f} is either
in tail position, or followed by a tail call to another cps function
whose parameters are \emph{non-shared} variables, that cannot be
modified by \texttt{f}.  This restriction about shared variables is
ensured by the \emph{boxing} pass detailed in Section~\ref{sec:boxing}.

\begin{definition}[Extruded and shared variables]\label{def:extruded}
\emph{Extruded variables} are local variables (or function parameters)
the address of which has been retained using the ``address of'' operator
(\verb|&|).

\emph{Shared variables} are either extruded or global variables.\qed
\end{definition}
Thus, the set of shared variables includes every variable that might be modified
by several functions called in the same dynamic chain.

\begin{definition}[CPS-convertible form]
    A function {\tt h} is in \emph{CPS-convertible form} if every call to
    a cps function that it contains matches one of the following
    patterns, where both \texttt{f} and \texttt{g} are cps functions,
    \texttt{a$_\text{\tt 1}$, ..., a$_\text{\tt n}$} are any C expressions and
    \texttt{x, y$_\text{\tt 1}$, ..., y$_\text{\tt n}$} are non-shared
    variables:
    \begin{align}
        \text{\tt return f(e$_\text{\tt 1}$, \ldots, e$_\text{\tt n}$);}\\
        \text{\tt x = f(e$_\text{\tt 1}$, \ldots, e$_\text{\tt n}$); return g(x,
            y$_\text{\tt 1}$, \ldots, y$_\text{\tt n}$);}\\
        \text{\tt f(e$_\text{\tt 1}$, \ldots, e$_\text{\tt n}$); return g(x,
            y$_\text{\tt 1}$, \ldots, y$_\text{\tt n}$);}\label{useless1}\\
        \text{\tt f(e$_\text{\tt 1}$, \ldots, e$_\text{\tt n}$); return;}\label{useless2}\\
        \text{\tt f(e$_\text{\tt 1}$, \ldots, e$_\text{\tt n}$); g(x,
            y$_\text{\tt 1}$, \ldots, y$_\text{\tt n}$); return;}\label{useless3}\\
        \text{\tt x = f(e$_\text{\tt 1}$, \ldots, e$_\text{\tt n}$); g(x,
            y$_\text{\tt 1}$, \ldots, y$_\text{\tt n}$);
            return;}\label{useless4}
    \end{align}
    \qed
\end{definition}
We use \texttt{return} to explicitly mark calls in tail position.  The
forms (\ref{useless1}) to (\ref{useless4}) are only necessary to
handle the cases where \texttt{f} and \texttt{g} return \texttt{void};
in the rest of this paper, we ignore these cases and focus on the
essential cases (1) and (2).

\paragraph{Early evaluation of non-shared variables}
To understand why the definition of CPS\hyp{}convertible form requires non-shared
variables for the parameters of \texttt{g}, consider what happens when
converting a CPS-convertible function.  We use a partial CPS conversion, as
explained above, focused on tail positions.  In a functional language, we would
define the CPS conversion as follows, where $\text{\tt f}^\star$ is
the CPS-converted form of \texttt{f} and \texttt{k} the current continuation:
\begin{center}
{\tt
\begin{tabular}{lcl}
    return a; &$\rightarrow$& return k(a);\vspace{1.2\jot}\\
    return f(a$_\text{\tt 1}$, \ldots, a$_\text{\tt n}$);
    &$\rightarrow$&
    return f$^\star$(a$_\text{\tt 1}$, \ldots, a$_\text{\tt n}$, k);\vspace{1.2\jot}\\
    x = f(a$_\text{\tt 1}$, \ldots, a$_\text{\tt n}$);&
    \multirow{2}{*}{$\rightarrow$}&
    k' = $\lambda$x.\  g$^\star$(x, y$_\text{\tt 1}$, \ldots, y$_\text{\tt n}$, k);
    \\
    \quad return g(x, y$_\text{\tt 1}$, \ldots, y$_\text{\tt n}$);
    &&
    \quad return f$^\star$(a$_\text{\tt 1}$, \ldots, a$_\text{\tt n}$, k');
\end{tabular}
}
\end{center}
Note the use of the lambda-abstraction
$\lambda\text{\tt x. g$^\star$(x, y$_\text{\tt 1}$, ..., y$_\text{\tt n}$, k)}$
to represent the continuation of the call to \texttt{f$^\star$} in the last
case.  In a call-by-value language, this continuation can be described as:
\emph{``get the return value of\/ {\tt f}$^\star$, bind it to {\tt x}, evaluate
the variables {\tt y$_\text{\tt 1}$} to {\tt y$_\text{\tt n}$} and call {\tt
g}$^\star$ with these parameters.''}

Representing this continuation in C raises the problem of evaluating the values
of \texttt{y$_\text{\tt 1}$} to \texttt{y$_\text{\tt n}$} after the call to {\tt
    f}$^\star$ has completed: these variables are local to the enclosing CPS-converted
function and, as such, allocated in a stack frame which will be discarded
when it returns.  To keep them available until the evaluation of the
continuation, we would need to store them in an environment and garbage-collect
them at some point.  We want to avoid this approach as much as possible for
performance reasons since it implies an extra layer of indirection and extra
bookkeeping.

We use instead a technique that we call \emph{early evaluation} of variables.
It is based on the following property: if we can ensure that the variables
\texttt{y$_\text{\tt i}$} cannot be modified by the function \texttt{f}, then it
is correct to commute their evaluation with the call to \texttt{f}.  Because the
CPC translator produces code where these variables are not shared, thanks to the
boxing pass (Section~\ref{sec:boxing}), it is indeed guaranteed that they cannot
be modified by a call to another function in the same dynamic chain.  In CPC,
the CPS conversion therefore evaluates these variables when the
continuation is built, before calling \texttt{f}, and stores their values
directly in the continuation.

We finally define the CPS conversion pass, using early evaluation and
the fact mentioned above that continuations are abstract data structures
manipulated by two operators: \texttt{push}, which adds a function to
the continuation, and \texttt{invoke} which executes a continuation.
\begin{definition}[CPS conversion]
    \label{def:cps-conversion}
    The \emph{CPS conversion} translates the tail positions of every
    CPS-convertible term as follows, where $\text{\tt f}^\star$ is the
    CPS-converted form of \texttt{f} and \texttt{k} is the current
    continuation:
\begin{center}
{\tt
\begin{tabular}{lcl}
    return a; &$\rightarrow$& return invoke(k, a);\vspace{1.2\jot}\\
    \multirow{2}{*}{return f(a$_\text{\tt 1}$, \ldots, a$_\text{\tt n}$);}
    &\multirow{2}{*}{$\rightarrow$}&
    push (k, f$^\star$, a$_\text{\tt 1}$, \ldots, a$_\text{\tt n}$);\\
    &&\quad return invoke(k);\vspace{1.2\jot}\\
    \multirow{3}{*}{\begin{tabular}{@{}l}x = f(a$_\text{\tt 1}$, \ldots, a$_\text{\tt n}$);\\
            \quad return g(x, y$_\text{\tt 1}$, \ldots, y$_\text{\tt n}$);
 \end{tabular}}
 &&
 push (k, g$^\star$, y$_\text{\tt 1}$, \ldots, y$_\text{\tt n}$);\\
 &$\rightarrow$&
 \quad push (k, f$^\star$, a$_\text{\tt 1}$, \ldots, a$_\text{\tt n}$);\\
 &&
 \quad return invoke(k);
\end{tabular}
}
\end{center}
\qed
\end{definition}
A proof of correctness of this conversion, including a proof that early
evaluation is correct in the absence of shared variables, is available in the
technical report \cite{kerneis12}.

\paragraph{Implementation}
From an implementation standpoint, the continuation {\tt k} is a stack of
function pointers and parameters, and \texttt{push} adds elements to
this stack.  The function \texttt{invoke} calls the first function of
the stack, with the rest of the stack as its parameters.  The form
\texttt{invoke(k, a)} also takes care of passing the value \texttt{a} to
the continuation \texttt{k}; it simply pushes {\tt a} at the top of the
continuation before calling its first function.

The current implementation of \texttt{push} grows continuations by a
multiplicative factor when the array is too small to contain the pushed
parameters, and never shrinks them.  While this might in principle
waste memory in the case of many long-lived continuations with an
occasional deep call stack, we believe that this case is rare enough not
to justify complicating the implementation with this optimisation.

The function \texttt{push} does not align parameters on word boundaries,
which leads to smaller continuations and easier store and load
operations.  Although word-aligned reads and writes are more efficient
in general, our tests showed little or no impact in the CPC programs we
experimented with, on x86 and x86-64 architectures: the worst case has
been a 10\,\% slowdown in a micro-benchmark with deliberately misaligned
parameters.  We have reconsidered this trade-off when porting CPC to
MIPS, an architecture with no hardware support for unaligned accesses,
and added a compilation option to align continuation frames.  However
we keep this option disabled by default, until we have more experimental
data to assess its efficiency.

There is one difference between Definition~\ref{def:cps-conversion} and
the implementation.  In a language with proper tail calls, each function
would simply invoke the continuation directly; in C, which is not
properly tail-recursive, doing that leads to unbounded growth of the
native call stack.  Therefore, the tail call \texttt{return invoke(k)}
cannot be used; we work around this issue by using a ``trampolining''
technique \cite{ganz}.  Instead of calling \texttt{invoke}, each cps
function returns its continuation: \texttt{return k}.  The main
event loop iteratively receives these continuations and invokes them
until a CPC primitive returns \texttt{NULL}, which yields to the next
CPC thread.

In the following sections, we show how the boxing, splitting and
lambda-lifting passes translate any CPC program into CPS-convertible
form.

\subsection{Boxing}\label{sec:boxing}

Boxing is a common, straightforward technique to encapsulate mutable variables.
It is necessary to ensure the correctness of the CPS conversion and
lambda-lifting passes (Sections~\ref{sec:cps-conversion}
and~\ref{sec:lifting-intro}).  However, boxing induces an expensive indirection
to access boxed variables.  In order to keep this cost at an acceptably low
level, we box only a subset of all variables --- namely extruded variables, whose
address is retained with the ``address of'' operator \verb+&+ (see
Definition~\ref{def:extruded}, p.~\pageref{def:extruded}).

\paragraph{Example} Consider the following function:
\begin{verbatim}
cps void f(int x) {
  int y = 0;
  int *p1 = &x, *p2 = &y;
  /* ... */
  return;
}
\end{verbatim}
The local variables {\tt x} and {\tt y} are extruded, because their
address is stored in the pointers {\tt p1} and {\tt p2}.  The boxing
pass allocates them on the heap at the beginning of the function {\tt f}
and frees them before {\tt f} returns.  For instance, in the next program every
occurrence of {\tt x} has been replaced by a pointer indirection {\tt *px}, and
{\tt \&x} by the pointer {\tt px}.
\begin{verbatim}
cps void f(int x) {
  int *px = malloc(sizeof(int));
  int *py = malloc(sizeof(int));
  *px = x;                 /* Initialise px */
  *py = 0;
  int *p1 = px, *p2 = py;
  /* ... with x and y replaced accordingly */
  free(px); free(py);
  return;
}
\end{verbatim}
The extruded variables {\tt x} and {\tt y} are not used anymore (except
to initialise {\tt px}).  Instead, {\tt px} and {\tt py} are used; note
that these variables, contrary to {\tt x} and {\tt y}, are not extruded:
they hold the address of other variables, but their own address is not
retained.  After the boxing pass, there are no more extruded variables
used in cps functions.

\paragraph{Cost analysis}

The efficiency of CPC relies in great part on avoiding boxing as much as
possible.  Performance-wise, we expect boxing only extruded
variables to be far less expensive than boxing every lifted variable.
Indeed, in a typical C program, few local variables have their address
retained compared to the total number of variables.

Experimental data confirm this intuition: in Hekate, the CPC translator
boxes 13 variables out of 125 lifted parameters.  This result is
obtained when compiling Hekate with the current CPC implementation.  They
take into account the fact that the CPC translator tries to be smart
about which variables should be boxed or lifted: for instance, if the
address of a variable is retained with the ``address of'' operator
\verb|&| but never used, this variable is not considered as extruded.
Using a naive implementation, however, does not change the proportion of
boxed variables: with optimisations disabled, 29 variables are boxed out
of 323 lifted parameters.  In both cases, CPC boxes about 10\,\% of the
lifted variables, which is an acceptable overhead.

\paragraph{Interaction with other passes}
The boxing pass yields a program without the ``address of'' operator
({\tt \&}); extruded variables are allocated on the heap, and only
pointers to them are copied by the lambda-lifting and CPS-conversion
passes rather than extruded variables themselves.  One may wonder,
however, whether it is correct to perform boxing before every
other transformation.  It turns out that boxing does not interfere with
the other passes, because they do not introduce any additional
``address of'' operators.  The program therefore remains free of
extruded variables.  Moreover, it is preferable to box early, before
introducing inner functions, since it makes it easier to identify the
entry and exit points of the original function, where variables are
allocated and freed.

\paragraph{Extruded variables and tail calls}
Although we keep the cost of boxing low, with about 10\,\% of boxed
variables, boxing has another, hidden cost: it breaks tail recursive cps
calls.  Since the boxed variables might, in principle, be used during
the recursive calls, one cannot free them beforehand.  Therefore,
functions featuring extruded variables do not benefit from the automatic
elimination of tail recursive calls induced by the CPS conversion.
While this prevents CPC from optimising tail recursive calls ``for
free'', it is not a real limitation: the C language does not ensure
the elimination of tail recursive calls anyway, as the stack frame
should similarly be preserved in case of extruded variables, and C
programmers are used not to rely on it.

\subsection{Splitting}\label{sec:splitting}

To transform a CPC program into CPS-convertible form, the CPC
translator needs to ensure that every call to a cps function is either
in tail position or followed by a tail call to another cps function.  In
the original CPC code, calls to cps functions might, by chance, respect
this property but, more often than not, they are followed by some
direct-style (non-cps) code.  The role of the CPC translator is therefore to
replace this direct-style chunk of code following a cps call by a tail call
to a cps function which encapsulates this chunk.  We call this pass
\emph{splitting} because it splits each original cps function into many,
mutually recursive, cps functions in CPS-convertible form.

To reach CPS-convertible form, the splitting pass must introduce tail calls
after every existing cps call.  This is done in two steps: we first introduce
a \texttt{goto} after every existing cps call (Section~\ref{sec:explicit-flow}),
then we translate these \texttt{goto} into tail calls
(Section~\ref{sec:goto-elimination}).

The first step consists in introducing a \texttt{goto} after every cps call.  Of
course, to keep the semantics of the program unchanged, this inserted
\texttt{goto} must jump to the statement following the tail call in the
control-flow graph: it makes the control flow explicit, and prepares the second
step which translates these \texttt{goto} into tail calls.  In most cases, the
control flow falls through linearly and inserting \texttt{goto} statements is
trivial; as we shall see in Section~\ref{sec:explicit-flow}, more care is
required when the control flow crosses loops and labelled blocks.  This step
produces code which resembles CPS-convertible form, except that it uses
\texttt{goto} instead of tail calls.

The second step is based on the observation that tail calls are equivalent to
jumps.  We convert each labelled block into an inner function, and each
\texttt{goto} statement into a tail call to the associated function, yielding a
program in CPS-convertible form.

We detail these two steps in the rest of this section.

\subsubsection{Explicit flow of control}\label{sec:explicit-flow}

When a cps call is not in tail position, or followed by a tail cps call, the CPC
translator adds a \texttt{goto} statement after it, to jump explicitly to the
next statement in the control-flow graph.  These inserted \texttt{goto} are to
be converted into tail cps calls in the next step.

In most cases, the control flow falls through linearly, and making it explicit
is trivial.  For instance,
\begin{verbatim}
cpc_yield(); rc = 0;
\end{verbatim}
becomes
\begin{verbatim}
cpc_yield(); goto l;
l: rc = 0;
\end{verbatim}
However, loops and conditional jumps require more care:
\begin{verbatim}
while(!timeout) {
  int rc = cpc_read();
  if(rc <= 0) break;
  cpc_write();
}
reset_timeout();
\end{verbatim}
is converted to:
\begin{verbatim}
while_label:
  if(!timeout) {
    int rc = cpc_read(); goto l;
    l:
      if(rc <= 0) goto break_label;
      cpc_write(); goto while_label;
  }
break_label:
  reset_timeout();
\end{verbatim}
More generally, when the flow of control after a cps call goes outside
of a loop (\texttt{for}, \texttt{while} or \texttt{do ... while})
or a \texttt{switch} statement, the CPC translator simplifies these
constructs into \texttt{if} and \texttt{goto} statements, adding the
necessary labels on the fly.

Although adding trivial \texttt{goto}, and making loops explicit, brings the
code in a shape close to CPS-convertible form, we need to add some more
\texttt{goto} statements for the second step to correctly encapsulate chunks of
code into cps functions.  Consider for instance the following piece of code:
\begin{verbatim}
if(rc < 0) {
  cpc_yield(); rc = 0;
}
printf("rc = %d\n", rc);
return rc;
\end{verbatim}
With the rules described above, it would be translated to:
\begin{verbatim}
if(rc < 0) {
  cpc_yield(); goto l;
  l: { rc = 0; }
}
printf("rc = %d\n", rc);
return rc;
\end{verbatim}
This is not enough because the block labelled by \texttt{l} will be
converted to a cps function in the second step.  But if we encapsulate only the
\texttt{\{rc = 0;\}} part, we will miss the call to \texttt{printf} when
\texttt{rc < 0}.  Therefore, to ensure a correct encapsulation in the
second step, we also need to make the flow of control explicit when
exiting a labelled block.  The example becomes:
\begin{verbatim}
if(rc < 0) {
  cpc_yield(); goto l;
  l: { rc = 0; goto done; }
}
done:
  printf("rc = %d\n", rc);
  return rc;
\end{verbatim}

After this step, every cps call is either in tail position or followed
by a \texttt{goto statement}, and every labelled block exits with either
a \texttt{return} or a \texttt{goto} statement.

\subsubsection{Eliminating gotos}\label{sec:goto-elimination}

It is a well-known fact in the compiler folklore that a tail call is
equivalent to a goto.  It is perhaps less known that a goto is equivalent
to a tail call \cite{Ste76,wijngaarden}: the block of any destination label
\texttt{l} is encapsulated in an inner function \texttt{l()}, and each
\texttt{goto l;} is replaced by a tail call to \texttt{l}.

Coming back to the first example of Section~\ref{sec:explicit-flow}, the
label \texttt{l} yields a function \texttt{l()}:
\begin{verbatim}
f(); return l();
cps void l() { rc = 0; }
\end{verbatim}
We see that the first line is now in CPS-convertible form.  Note
again that we use \texttt{return} to mark tail calls explicitly in the C
code.

Applying the same conversion to loops yields mutually recursive
functions.  For instance, the \texttt{while} loop in the second example
is converted into \texttt{while\_label()}, \texttt{l()} and
\texttt{break\_label()}:
\begin{verbatim}
while_label();
cps void while_label() {
  if(!timeout) {
    int rc = cpc_read(); return l();
    cps void l() {
      if(rc <= 0) return break_label();
      cpc_write(); return while_label();
    }
  }
}
cps void break_label() {
  reset_timeout();
}
\end{verbatim}
When a label, like \texttt{while\_label} above, is reachable not only
through a \texttt{goto}, but also directly through the linear flow of
the program, it is necessary to call its associated function at the
point of definition; this is what we do on the first line of this
example.

Note that splitting may introduce free variables; for instance, in the
previous example, \texttt{rc} is free in the function \texttt{l}.  In
this intermediate C-like code, inner functions are considered just like any
other C block when it comes to the scope and lifespan of variables: each
variable lives as long as the block or function which defines it, and it
can be read and modified by any instruction inside this block, including
inner functions.  There is in particular no copy of free variables in
inner functions; variables are shared inside their block.   In the
previous example, the lifespan of the variable \texttt{rc} is the
\texttt{if} block, including the call to the function \texttt{l} which
reads the free variable allocated in its enclosing function
\texttt{while\_label}.

The third example shows the importance of having explicit flow of
control at the end of labelled blocks.  After \texttt{goto}
elimination, it becomes:
\begin{verbatim}
if(rc < 0) {
  cpc_yield(); return l();
  cps int l() { rc = 0; return done(); }
}
cps int done() {
  printf("rc = %d\n", rc);
  return rc;
}
return done();
\end{verbatim}
Note the tail call to \texttt{done} at the end of \texttt{l} to ensure
that the \texttt{printf} is executed when \texttt{rc < 0}, and again the
call on the last line to execute it otherwise.

\subsubsection{Implementation}

The actual implementation of the CPC translator implements splitting on
an as-needed basis: a given subtree of the abstract syntax tree (AST) is
only transformed if it contains CPC primitives that cannot be
implemented in direct style.  Our main concern here is to transform the
code as little as possible, on the assumption that the gcc compiler is
optimised for human-written code.

To perform splitting, the CPC translator iterates over the AST, checking
whether the cps functions are in CPS-convertible form and
interleaving the two steps described above to reach CPS-convertible form
incrementally.  On each pass, when the translator finds a cps call it
dispatches on the statement following it:
\begin{itemize}
  \item in the case of a {\tt return}, the fragment is already CPS-convertible
  and the translator continues;
  \item in the case of a {\tt goto}, it is converted into a tail call,
      with the corresponding label turned into an inner function
      (Section~\ref{sec:goto-elimination}), and the translator starts
      another pass;
  \item for any other statement, a {\tt goto} is added to make the flow
      of control explicit, converting enclosing loops too if necessary
      (Section~\ref{sec:explicit-flow}).  The translator then starts
      another pass and will eventually convert the introduced {\tt goto}
      into a tail call.
\end{itemize}

At the end of the splitting pass, the translated program is in
CPS-convertible form.  However, it is not quite ready for CPS conversion
because we introduced inner functions, which makes it invalid C.  In
particular, these functions may contain free variables.  The next pass,
lambda-lifting, takes care of these free variables to get a valid C
program in CPS-convertible form, suitable for the CPS conversion pass
described in Section~\ref{sec:cps-conversion}.

\subsection{Lambda-lifting}\label{sec:lifting-intro}

Lambda-lifting \cite{johnsson} is a standard technique to eliminate free
variables in functional languages.  It proceeds in two phases
\cite{danvy}.  In the first pass (``parameter lifting''), free variables
are replaced by local, bound variables.  In the second pass (``block
floating''), the resulting closed functions are floated to top-level.

\paragraph{Example}
Coming back to the latest example, we add an enclosing function
\texttt{f} to define \texttt{rc} and make the fragment self-contained:
\begin{verbatim}
cps int f(int rc) {
  if(rc < 0) {
    cpc_yield(); return l();
    cps int l() { rc = 0; return done(); }
  }
  cps int done() {
    printf("rc = %d\n", rc);
    return rc;
  }
  return done();
}
\end{verbatim}
The function \texttt{f} contains two inner functions, \texttt{l} and
\texttt{done}.  The local variable \texttt{rc} is used as a free
variable in both of these functions.

Parameter lifting consists in adding the free variable as a parameter of
every inner function:
\begin{verbatim}
cps int f(int rc) {
  if(rc < 0) {
    cpc_yield(); return l(rc);
    cps int l(int rc) { rc = 0; return done(rc); }
  }
  cps int done(int rc) {
    printf("rc = %d\n", rc);
    return rc;
  }
  return done(rc);
}
\end{verbatim}
Note that \texttt{rc} is now a parameter of \texttt{l} and
\texttt{done}, and has been added accordingly whenever these functions
where called.  There are, now, three copies of \texttt{rc};
alpha-conversion makes this more obvious:
\begin{verbatim}
cps int f(int rc1) {
  if(rc1 < 0) {
    cpc_yield(); return l(rc1);
    cps int l(int rc2) { rc2 = 0; return done(rc2); }
  }
  cps int done(int rc3) {
    printf("rc = %d\n", rc3);
    return rc3;
  }
  return done(rc1);
}
\end{verbatim}
Once the parameter lifting step has been performed, there is no free variable
anymore and the block floating step consists in extracting these inner,
closed functions at top-level:
\begin{verbatim}
cps int f(int rc1) {
  if(rc1 < 0) {
    cpc_yield(); return l(rc1);
  }
  return done(rc1);
}
cps int l(int rc2) {
  rc2 = 0;
  return done(rc2);
}
cps int done(int rc3) {
  printf("rc = %d\n", rc3);
  return rc3;
}
\end{verbatim}
Applying boxing, splitting and lambda-lifting always yields
CPS-convertible programs:
\begin{itemize}
    \item every call to a cps function is either a tail call (not
        affected by the transformations) or followed by a tail cps call
        (introduced in the splitting pass),
    \item the parameters of this second cps call are local variables,
        since they are introduced by the lambda-lifting pass,
    \item these parameters are not shared because they are neither
        global (local variables) nor extruded (the boxing pass
        ensures that there are no more extruded variables in the
        program).
\end{itemize}

\paragraph{Lambda-lifting in imperative languages}
There is one major issue with applying lambda-lifting to C
extended with inner functions: this transformation is in general not
correct in a call-by-value languages with mutable variables.

Consider what would happen if splitting were to yield the following code:
\begin{verbatim}
cps int f(int rc) {
  cps void set() { rc = 0; return; }
  cps void done() {
    printf("rc = %d\n", rc);
    return;
  }
  set(); return done();
}
\end{verbatim}
This code, which is in CPS-convertible form, prints out \texttt{rc = 0}
whatever the original value of \texttt{rc} was: the call to \texttt{set}
sets \texttt{rc} to 0 and the call to \texttt{done} displays it.

Once lambda-lifted, this code becomes:
\begin{verbatim}
cps int f(int rc1) {
  set(rc1); return done(rc1);
}
cps void set(int rc2) {
  rc2 = 0;
  return;
}
cps void done(int rc3) {
  printf("rc = %d\n", rc3);
  return;
}
\end{verbatim}
The result has changed: the code now displays the original value of
\texttt{rc} passed to \texttt{f} rather than 0.  The reason why
lambda-lifting is incorrect in that case is because \texttt{set} and
\texttt{done} work on two separate copies \texttt{rc}, \texttt{rc2} and
\texttt{rc3}, whereas in the original code there was only one instance
of \texttt{rc} shared by all inner functions.

This issue is triggered by the fact that the function \texttt{set} is
not called in tail position.  This non-tail call allows us to observe
the incorrect value of \texttt{rc3} after \texttt{set} has modified the
copy \texttt{rc2} and returned.  If \texttt{set} were called in tail
position instead, the fact that it operates on a copy of \texttt{rc1}
would remain unnoticed.

\paragraph{Lambda-lifting and tail calls}
In fact, although lambda-lifting is not correct in general in a
call-by-value language with mutable variables, it becomes correct once
restricted to functions called in tail position and non-extruded
parameters.
More precisely, in the absence of extruded variables, it is safe to
lift a parameter provided every inner function where this parameter is
lifted is only called in tail position.  We show this result in
Section~\ref{sec:lifting} (Theorem~\ref{thm:lambda-lifting-correctness},
p.~\pageref{thm:lambda-lifting-correctness}).

Inner functions in CPC are the result of \texttt{goto} elimination
during the splitting step.  As a result, they are always called in tail
position.  Moreover, as explained in Section~\ref{sec:boxing}, the
boxing pass ensures that extruded variables have been eliminated at
this point of the transformation.  Hence, lambda-lifting is correct in
the case of CPC.

\paragraph{Implementation}  To lift as few variables as possible,
lambda-lifting is implemented incrementally.  Rather than lifting every
parameter, the translators looks for free variables to be lifted in a
cps function and adds them as parameters at call points; this creates
new free variables, and the translator iterates until it reaches a fixed
point.

This implementation might be further optimised with a liveness analysis,
which would in particular avoid lifting  uninitialised parameters.  The
current translator performs a very limited analysis: only variables
which are used (hence alive) in a single function are not lifted.

\paragraph{Experimental results}
The common technique to use lambda-lifting in an imperative language is
to box every mutated variable, in order to duplicate pointers to these
variables instead of the variables themselves.  To quantify the amount
of boxing avoided by our technique of lambda-lifting tail-called
functions, we used a modified version of CPC which blindly boxes every
lifted parameter and measure the amount of boxing that it induced in
Hekate, the most substantial program written with CPC so far.

Hekate contains 260 local variables and function parameters, spread
across 28 cps functions\footnote{These numbers leave out direct-style
    functions, which do not need to be converted, and around 200 unused
    temporary variables introduced by a single pathological
    macro-expansion in the \textit{curl} library.}.  Among them, 125
variables are lifted.  A naive lambda-lifting pass would therefore need
to box almost 50\,\% of the variables.

On the other hand, boxing extruded variables only carries a much smaller
overhead: in Hekate, the current CPC translator boxes only 5\,\% of the
variables in cps functions.  In other words, 90\,\%  of the lifted
variables in Hekate are safely left unboxed, keeping the overhead
associated with boxing to a reasonable level.

\section{Lambda-lifting in an imperative language}
\label{sec:lifting}

To prove the correctness of lambda-lifting in an imperative,
call-by-value language when functions are called in tail position, we do
not reason directly on CPC programs, because the semantics of C is too
broad and complex for our purposes.  The CPC translator leaves most
parts of converted programs intact, transforming only control structures
and function calls.  Therefore, we define a simple language with
restricted values, expressions and terms, that captures the features we
are most interested in (Section~\ref{sec:definitions}).

The reduction rules for this language (Section~\ref{sec:naive-def}) use
a simplified memory model without pointers and enforce that local
variables are not accessed outside of their scope, as ensured by our
boxing pass.  This is necessary since we have seen that lambda-lifting
is not correct in general in the presence of extruded variables.

It turns out that the ``naive'' reduction rules defined in
Section~\ref{sec:naive-def} do not provide strong enough invariants to
prove this correctness theorem by induction, mostly because we represent
memory with a store that is deeply modified by lambda-lifting.
Therefore, in Section~\ref{sec:semopt}, we define an equivalent,
``optimised'' set of reduction rules which enforces more regular stores
and closures.

The proof of correctness is then carried in
Section~\ref{sec:correction-ll} using these optimised rules.  We first
define the invariants needed for the proof and formulate a strengthened
version of the correctness theorem (Theorem~\ref{thm:correction-ll},
Section~\ref{sec:strong-invariants}).  A comprehensive overview of the
proof is then given in Section~\ref{sec:overview}.  The proof is fully
detailed in Section~\ref{sec:proof-correctness}, with the help of a few
lemmas to keep the main proof shorter
(Sections~\ref{sec:rewriting-lemmas} and~\ref{sec:aliasing-lemmas}).

The main limitation of this proof is that
Theorems~\ref{thm:lambda-lifting-correctness}
and~\ref{thm:correction-ll} are implications, not equivalences: we do
not prove that if a term does not reduce, it will not reduce once
lifted.  For instance, this proof does not ensure that lambda-lifting
does not break infinite loops.

\subsection{Definitions}
\label{sec:definitions}

In this section, we define the terms
(Definition~\ref{def:full-language}), the reduction rules
(Section~\ref{sec:naive-def}) and the lambda-lifting transformation
itself (Section~\ref{sec:lifting-def}) for our small imperative
language.  With these preliminary definitions, we are then able to
characterise \emph{liftable parameters}
(Definition~\ref{dfn:var-liftable-simple}) and state the main
correctness theorem (Theorem~\ref{thm:lambda-lifting-correctness},
Section~\ref{sec:correctness}).

\begin{definition}[Values, expression and terms]\label{def:full-language}

Values are either boolean and integer constants or $\unit$, a special
value for functions returning \texttt{void}.
\[v \Coloneqq \quad\unit \;|\; \true \;|\; \false \;|\; n \in \mathbf{N}\]

Expressions are either values or variables.  We deliberately omit
arithmetic and boolean operators, with the sole concern of avoiding
boring cases in the proofs.
\[\expr \Coloneqq \quad v \;|\; x\;|\; \dotsc\]

Terms are made of assignments, conditionals, sequences, recursive
functions definitions and calls.
    \begin{align*}
    T \Coloneqq & \expr
        \;|\; x \coloneqq T
        \;|\; \ite{T}{T}{T}
        \;|\; T\ ;\ T \\
        \;|\; & \letrec{f(\range{x}{1}{n})}{T}{T}
        \;|\; f(T,\dotsc,T)
        \end{align*}
        \qed
\end{definition}
Our language focuses on the essential details affected by the
transformations: recursive functions, conditionals and memory accesses.
Loops, for instance, are ignored because they can be expressed in terms
of recursive calls and conditional jumps --- and that is, in fact, how
the splitting pass translates them (Section~\ref{sec:splitting}).
Since lambda-lifting happens after the splitting pass, our language
need to include inner functions (although they are not part of the C
language), but it can safely exclude \texttt{goto} statements.

One important simplification of this language compared to C is the lack of
pointers.  However, remember that we are lifting only local, stack-allocated
variables, and that these variables cannot be accessed outside of their scope,
as ensured by our boxing pass.  Since we get rid of the ``address of'' operator
\verb+&+, pointers remaining in CPC code after boxing always point to the heap,
never to the stack.  Adding a heap and pointers to our language would only make
it larger without changing the proof of correctness.

\subsubsection{Naive reduction rules\label{sec:naive-def}}

\paragraph{Environments and stores}
Handling inner functions requires explicit closures in the reduction
rules.  We need environments, written $\rho$, to bind variables to
locations, and a store, written $s$, to bind locations to values.

\emph{Environments} and \emph{stores} are partial functions, equipped
with a single operator which extends and modifies a partial function:
\ajout{\cdot}{\cdot}{\cdot}.

\begin{definition} The modification (or extension) $f'$ of a partial
function $f$, written $f' = \ajout{f}{x}{y}$, is defined as follows:
\begin{align*}
f'(t) =& \begin{cases}
y&\text{when $t$ = $x$}\\
f(t)&\text{otherwise}
\end{cases}\\
\dom(f') =& \dom(f)\cup\{x\}
\end{align*}
\qed
\end{definition}

\begin{definition}[Environments of variables and functions]
Environments of variables are defined inductively by
\[\rho \Coloneqq \varepsilon \;|\; (x,l)\concat\rho,\]
i.e.\ the empty domain function and $\ajout{\rho}{x}{l}$ (respectively).

Environments of functions, on the other hand, associate function names
to closures:
\[\F : \{f, g, h, \dotsc\} \rightarrow
    \{\fun{\range{x}{1}{n}}{T}{\rho,\F}\}.\]
\qed
\end{definition}

Note that although we have a notion of locations, which correspond
roughly to memory addresses in C, there is no way to copy, change or
otherwise manipulate a location directly in the syntax of our language.
This is on purpose, since adding this possibility would make
lambda-lifting incorrect: it translates the fact, ensured by the boxing
pass in the CPC translator, that there are no extruded variables in the
lifted terms.

\paragraph{Reduction rules}
We use classical big-step reduction rules for our language
(Figure~\ref{sem-proof:naive}, p.~\pageref{sem-proof:naive}).

\begin{figure}
\begin{gather*}
  \inferrule*[Left=(val)]{ }{\reductionN{v}{s}{v}{s}} \qquad\qquad
  \inferrule*[Left=(var)]{\rho\ x = l \in \dom\ s}{\reductionN{x}{s}{s\ l}{s}}\\
  \inferrule*[Left=(assign)]{\reductionN{a}{s}{v}{s'} \\
    \rho\ x = l \in \dom\ s'}{\
	\reductionN{x \coloneqq a}{s}{\unit}{\subst{s'}{l}{v}}}\qquad\qquad
  \inferrule*[Left=(seq)]{\reductionN{a}{s}{v}{s'} \\ \reductionN{b}{s'}{v'}{s''}}{\reductionN{a\ ;\ b}{s}{v'}{s''}}\\
  \inferrule*[Left=(if-t.)]{\reductionN{a}{s}{\true}{s'} \\ \reductionN{b}{s'}{v}{s''}}{\
	\reductionN{\ite{a}{b}{c}}{s}{v}{s''}}\qquad\qquad
  \inferrule*[Left=(if-f.)]{\reductionN{a}{s}{\false}{s'} \\ \reductionN{c}{s'}{v}{s''}}{\
	\reductionN{\ite{a}{b}{c}}{s}{v}{s''}}\\
  \inferrule*[Left=(letrec)]{\
    \reductionNF{b}{s}{v}{s'}{\mathcal{F'}} \\\\
    \mathcal{F'}=\ajout{\F}{f}{\fun{\range{x}{1}{n}}{a}{\rho,\F}}    }{\
	\reductionN{\letrec{f(\range{x}{1}{n})}{a}{b}}{s}{v}{s'}}\\
  \inferrule*[Left=(call)]{\
\F\,f = \fun{\range{x}{1}{n}}{b}{\rho',\mathcal{F'}}\\
\rho''= \drange{x}{l}{1}{n}\\
\text{$l_{i}$ fresh and distinct}\\\\
\forall i,\reductionN{a_i}{s_i}{v_i}{s_{i+1}} \\
\reductionNF[\rho''\concat\rho']{b}{\ajout{s_{n+1}}{l_i}{v_i}}{v}{s'}{\ajout{\mathcal{F'}}{f}{\F\,f}}
}{\
\reductionN{f(\range{a}{1}{n})}{s_{1}}{v}{s'}}
\end{gather*}
\caption{``Naive'' reduction rules\label{sem-proof:naive}}
\end{figure}

In the (call) rule, we need to introduce \emph{fresh} locations for the
parameters of the called function.  This means that we must choose
locations that are not already in use, in particular in the environments
$\rho'$ and $\F$.  To express this choice, we define two ancillary
functions, $\Env$ and $\Loc$, to extract the environments and locations
contained in the closures of a given environment of functions $\F$.
\begin{definition}[Set of environments, set of locations]
  \[\Env(\F) = \bigcup\left\{ \rho, \rho'
  \ |\ \fun{\range{x}{1}{n}}{M}{\rho,\mathcal{F'}} \in
  \Image(\F), \rho' \in \Env(\mathcal{F'})\right\}\]
  \[\Loc(\F) = \bigcup\left\{ \Image(\rho)
  \ |\ \rho \in
  \Env(\F)\right\}\]
  \[\text{A location $l$ is said to \emph{appear} in } \F \ssi l \in
      \Loc(\F).\]
  \qed
\end{definition}
These functions allow us to define fresh locations.
\begin{definition}[Fresh location] In the (call) rule, a location is
\emph{fresh} when:
\begin{itemize}
  \item $l \notin \dom(s_{n+1})$, i.e.\ $l$ is not already used in the store
  before the body of $f$ is evaluated, and
  \item $l$ doesn't appear in $\ajout{\mathcal{F'}}{f}{\F\,f}$, i.e.\
  $l$ will not interfere with locations captured in the environment of
  functions.
\end{itemize}
\qed
\end{definition}
Note that the second condition implies in particular that $l$ does not
appear in either $\F$ or $\rho'$.

\subsubsection{Lambda-lifting}\label{sec:lifting-def}

We mentioned in Section~\ref{sec:lifting-intro} that lambda-lifting can
be split into two parts: parameter lifting and block floating.  We will
focus only on the first part here, since the second one is trivial.
Parameter lifting consists in adding a free variable as a parameter of
every inner function where it appears free.  This step is repeated until
every variable is bound in every function, and closed functions can
safely be floated to top-level.  Note that although the transformation
is called lambda-lifting, we do not focus on a single function and try
to lift all of its free variables; on the contrary, we define the
lifting of a single free parameter $x$ in every possible function.

Usually, smart lambda-lifting algorithms strive to minimize the number
of lifted variables.  Such is not our concern in this proof: parameters
are lifted in every function where they might potentially be free.  (In
our implementation, the CPC translator actually uses a smarter approach to
avoid lifting too many parameters, as explained in
Section~\ref{sec:lifting-intro}.)

\begin{definition}[Parameter lifting in a term]\label{dfn:lifted-term}
Assume that $x$ is defined as a parameter of a given function $g$, and
that every inner function in $g$ is called $h_i$ (for some
$i\in\mathbf{N}$).  Also assume that function parameters are unique before
lambda-lifting.

\noindent Then, the \emph{lifted form} $\lift{M}$ of the term $M$ with
respect to $x$ is defined inductively as follows:
{\allowdisplaybreaks
  \begin{gather*}
  \lift{\unit} = \unit \qquad \lift{n} = n \\
  \lift{true} = true \qquad \lift{false} = false \\
  \lift{y} = y \quad \text{ and } \quad \lift{y \coloneqq a}= y \coloneqq
  \lift{a} \quad \text{(even if $y=x$)} \\
  \lift{a\ ;\ b} = \lift{a}\ ;\ \lift{b} \\
  \lift{\ite{a}{b}{c}} = \ite{\lift{a}}{\lift{b}}{\lift{c}} \\
  \lift{ \letrec{f(\range{x}{1}{n})}{a}{b} } =
  \begin{cases}
  \letrec{f(\range{x}{1}{n},x)}{\lift{a}}{\lift{b}}  &\text{if $f = h_i$}\\
  \letrec{f(\range{x}{1}{n})}{\lift{a}}{\lift{b}}   &\text{otherwise}
  \end{cases}\\
  \lift{f(\range{a}{1}{n})} =
  \begin{cases}
  f(\lift{a_{1}},\dotsc,\lift{a_{n}},x)&\text{if $f = h_i$ for some $i$}\\
  f(\lift{a_{1}},\dotsc,\lift{a_{n}})&\text{otherwise}
  \end{cases}
  \end{gather*}
}
  \qed
\end{definition}

\subsubsection{Correctness condition}\label{sec:correctness}

We claim that parameter lifting is correct for variables defined in
functions whose inner functions are called exclusively in \emph{tail
position}.  We call these variables \emph{liftable parameters}.

We first define tail positions as usual \cite{clinger}:
\begin{definition}[Tail position]
\emph{Tail positions} are defined inductively as follows:
\begin{enumerate}
\item $M$ and $N$ are in tail position in \ite{P}{M}{N}.
\item $N$ is in tail position in $N$ and $M \ ;\ N$ and \letrec{f(\range{x}{1}{n})}{M}{N}.
\end{enumerate}
\qed
\end{definition}
A parameter $x$ defined in a function $g$ is liftable if every inner
function in $g$ is called exclusively in tail position.
\begin{definition}[Liftable parameter] \label{dfn:var-liftable-simple}
A parameter $x$ is \emph{liftable} in $M$ when:
\begin{itemize}
\item $x$ is defined as the parameter of a function $g$,
\item inner functions in $g$, named $h_i$, are called exclusively in
tail position in $g$ or in one of the $h_i$.
\end{itemize}
\qed
\end{definition}
Our main theorem is that performing parameter-lifting on a liftable
parameter preserves the reduction:
\begin{theorem}[Correctness of lambda-lifting]
\label{thm:lambda-lifting-correctness}
If $x$ is a liftable parameter in $M$, then
\[\exists t,
\reductionNF[\varepsilon]{M}{\varepsilon}{v}{t}{\varepsilon} \text{ implies }
\exists t',
\reductionNF[\varepsilon]{\lift{M}}{\varepsilon}{v}{t'}{\varepsilon}.\]
\end{theorem}
Note that the resulting store $t'$ changes because lambda-lifting
introduces new variables, hence new locations in the store, and changes
the values associated with lifted variables;
Section~\ref{sec:correction-ll} is devoted to the proof of this theorem.
To maintain invariants during the proof, we need to use an equivalent,
``optimised'' set of reduction rules; it is introduced in the next section.

\subsection{Optimised reduction rules\label{sec:semopt}}

The naive reduction rules (Section~\ref{sec:naive-def}) are not
well-suited to prove the correctness of lambda-lifting.  Indeed, the
proof is by induction and requires a number of invariants on the
structure of stores and environments.  Rather than having a dozen of
lemmas to ensure these invariants during the proof of correctness, we
translate them as constraints in the reduction rules.

To this end, we introduce two optimisations --- minimal stores
(Section~\ref{sec:mini-store}) and compact closures
(Section~\ref{sec:compact-closures}) --- which lead to the definition of
an optimised set of reduction rules (Figure~\ref{sem-proof:opt},
Section~\ref{sec:opt-rules}).  The equivalence between optimised and
naive reduction rules is shown in the technical report~\cite{kerneis12}.

\subsubsection{Minimal stores} \label{sec:mini-store}

In the naive reduction rules, the store grows faster when reducing lifted
terms, because each function call adds to the store as many locations as it
has function parameters.  This yields stores of different sizes when
reducing the original and the lifted term, and that difference cannot be
accounted for locally, at the rule level.

Consider for instance the simplest possible case of lambda-lifting:
\begin{gather}
    \letrec{g(x)}{(\letrec{h()}{x}{h()})}{g(\unit)}\tag{original} \\
    \letrec{g(x)}{(\letrec{h(y)}{y}{h(x)})}{g(\unit)}\tag{lifted}
\end{gather}
At the end of the reduction, the store for the original term is
$\{l_x \mapsto \unit \}$ whereas the store for the lifted term is
$\{l_x \mapsto \unit ; l_y \mapsto \unit \}$.  More complex terms
would yield even larger stores, with many out-of-date copies of lifted
variables.

To keep the store under control, we need to get rid of useless variables
as soon as possible during the reduction.  It is safe to remove a
variable $x$ from the store once we are certain that it will never be
used again, i.e.\ as soon as the term in tail position in the function
which defines $x$ has been evaluated.  This mechanism is analogous
to the deallocation of a stack frame when a function returns.

To track the variables whose location can be safely reclaimed after the
reduction of some term $M$, we introduce \emph{split environments}.
Split environments are written $\env{\rhot}{\rho}$, where $\rhot$ is
called the \emph{tail environment} and $\rho$ the non-tail one; only the
variables belonging to the tail environment may be safely reclaimed.
The reduction rules build environments so that a variable $x$ belongs to
$\rhot$ if and only if the term $M$ is in tail position in the current
function $f$ and $x$ is a parameter of $f$.  In that case, it is safe to
discard the locations associated to all of the parameters of $f$,
including $x$, after $M$ has been reduced because we are sure that the
evaluation of $f$ is completed (and there is no first-class functions in
the language to keep references on variables beyond their scope of
definition).

We also define a  \emph{cleaning} operator, $\gc{\cdot}{\cdot}$, to
remove a set of variables from the store.
\begin{definition}[Cleaning of a store]
The store $s$ cleaned with respect to the variables in $\rho$, written
$\gc{\rho}{s}$, is defined as
$\gc{\rho}{s} = s |_{\dom(s)\setminus\Image(\rho)}$.
\qed
\end{definition}

\subsubsection{Compact closures} \label{sec:compact-closures}
Another source of complexity with the naive reduction rules is the inclusion of
useless variables in closures.  It is safe to remove from the
environments of variables contained in closures the variables that are also
parameters of the function: when the function is called, and the
environment restored, these variables will be hidden by the freshly
instantiated parameters.

This is typically what happens to lifted parameters: they are free
variables, captured in the closure when the function is defined, but
these captured values will never be used since calling the function adds
fresh parameters with the same names.  We introduce \emph{compact
closures} in the optimised reduction rules to avoid dealing with this
hiding mechanism in the proof of lambda-lifting.

A compact closure is a closure that does not capture any variable which
would be hidden when the closure is called because of function
parameters having the same name.
\begin{definition}[Compact closure and environment]
  A closure $\fun{\range{x}{1}{n}}{M}{\rho,\F}$ is a \emph{compact} closure
  if $\forall i, x_i\notin\dom(\rho)$ and \F\/ is compact.
  An environment is \emph{compact} if it contains only compact closures.
  \qed
\end{definition}
We define a canonical mapping from any environment $\F$ to a compact
environment $\close{\F}$, restricting the domains of every closure in
$\F$.
\begin{definition}[Canonical compact environment]
  The \emph{canonical compact environment} $\close{\F}$ is the
  unique environment with the same domain as $\F$ such that 
  \begin{align*}
    \forall f \in \dom(\F),
    \F\,f &= \fun{\range{x}{1}{n}}{M}{\rho,\mathcal{F'}}\\
    \text{implies }\close{\F}\ f &=
    \fun{\range{x}{1}{n}}{M}{\rho|_{\dom(\rho)\setminus\{\range{x}{1}{n}\}},\close{\mathcal{F'}}}.
  \end{align*}
    \qed
\end{definition}

\subsubsection{Optimised reduction rules} \label{sec:opt-rules}
Combining both optimisations yields the \emph{optimised} reduction rules
(Figure~\ref{sem-proof:opt}, p.~\pageref{sem-proof:opt}), used in
Section~\ref{sec:correction-ll} for the proof of lambda-lifting.

Consider for instance the rule (seq).
\[\inferrule*[Left=(seq)]{\reduction[\env{}{\rhot\concat\rho}]{a}{s}{v}{s'} \\
\reduction{b}{s'}{v'}{s''}}{\reduction{a\ ;\ b}{s}{v'}{s''}}\]
The environment of variables is split into the tail environment, $\rhot$, and the
non-tail one, $\rho$.  This means that $a\ ;\ b$ is in tail position in a function
whose parameters are the variables of $\rhot$. When we reduce the left part of
the sequence, $a$, we track the fact that it is not in tail position in this
function by moving $\rhot$ to the non-tail environment:
$\reduction[\env{}{\rhot\concat\rho}]{a}{s}{v}{s'}$.  On the other hand, when we
reduce $b$, we are in the tail of the term and the environment stays split.

As detailed above, we have introduced split environments in order to ensure
minimal stores.  Stores are kept minimal in the rules corresponding to tail
positions, the leaves of the reduction tree: (val), (var) and (assign).  In
these three rules, variables that appear in the tail environment are cleaned
from the resulting store: $\gc{\rhot}{s}$.

Finally, the (letrec) and (call) rules are modified to introduce compact
closures and split environments, respectively.  Compact closures are built in
the (letrec) rule by removing the parameters of $f$ from the captured
environment $\rho'$.  In the (call) rule, environments are split in a tail part,
which contains local variables of the called function, and a non-tail part,
which contains captured variables; only the former must be cleaned when the tail
instruction of the function is reduced.

\begin{figure}
\begin{gather*}
  \inferrule*[Left=(val)]{ }{\reduction{v}{s}{v}{\gc{\rhot}{s}}}\qquad\qquad
  \inferrule*[Left=(var)]{\rhot\concat\rho\ x = l \in \dom\ s}{\reduction{x}{s}{s\ l}{\gc{\rhot}{s}}}\\
  \inferrule*[Left=(assign)]{\reduction[\env{}{\rhot\concat\rho}]{a}{s}{v}{s'} \\
    \rhot\concat\rho\ x = l \in \dom\ s'}{\
	\reduction{x \coloneqq a}{s}{\unit}{\gc{\rhot}{\subst{s'}{l}{v}}}}\qquad\qquad
  \inferrule*[Left=(seq)]{\reduction[\env{}{\rhot\concat\rho}]{a}{s}{v}{s'} \\ \reduction{b}{s'}{v'}{s''}}{\reduction{a\ ;\ b}{s}{v'}{s''}}\\
  \inferrule*[Left=(if-t.)]{\reduction[\env{}{\rhot\concat\rho}]{a}{s}{\true}{s'} \\ \reduction{b}{s'}{v}{s''}}{\
	\reduction{\ite{a}{b}{c}}{s}{v}{s''}}\qquad\qquad
  \inferrule*[Left=(if-f.)]{\reduction[\env{}{\rhot\concat\rho}]{a}{s}{\false}{s'} \\ \reduction{c}{s'}{v}{s''}}{\
	\reduction{\ite{a}{b}{c}}{s}{v}{s''}}\\
  \inferrule*[Left=(letrec)]{\
    \reductionF{b}{s}{v}{s'}{\mathcal{F'}} \\\\
    \rho' = \rhot\concat\rho|_{\dom(\rhot\concat\rho)\setminus\{\range{x}{1}{n}\}} \\
    \mathcal{F'}=\ajout{\F}{f}{\fun{\range{x}{1}{n}}{a}{\rho',\F}}    }{\
	\reduction{\letrec{f(\range{x}{1}{n})}{a}{b}}{s}{v}{s'}}\\
  \inferrule*[Left=(call)]{\
\F\,f = \fun{\range{x}{1}{n}}{b}{\rho',\mathcal{F'}}\\
\rho''= \drange{x}{l}{1}{n}\\
\text{$l_{i}$ fresh and distinct}\\\\
\forall i,\reduction[\env{}{\rhot\concat\rho}]{a_i}{s_i}{v_i}{s_{i+1}} \\
\reductionF[\env{\rho''}{\rho'}]{b}{\ajout{s_{n+1}}{l_i}{v_i}}{v}{s'}{\ajout{\mathcal{F'}}{f}{\F\,f}}
}{\
\reduction{f(\range{a}{1}{n})}{s_{1}}{v}{\gc{\rhot}{s'}}}
\end{gather*}
\caption{Optimised reduction rules\label{sem-proof:opt}}
\end{figure}

\begin{theorem}[Equivalence between naive and optimised reduction rules]\label{thm:sem-equiv}
Optimised and naive reduction rules are equivalent: every reduction in
one set of rules yields the same result in the other.  It is necessary,
however, to take care of locations left in the store by the naive reduction:
  \[
  \reductionF[\env{\varepsilon}{\varepsilon}]{M}{\varepsilon}{v}{\varepsilon}{\varepsilon}
  \ssi
 \exists s,\reductionNF[\varepsilon]{M}{\varepsilon}{v}{s}{\varepsilon}
 \]
\end{theorem}
The proof of this theorem is detailed in the technical report~\cite{kerneis12}.

\subsection{Correctness of lambda-lifting\label{sec:correction-ll}}

In this section, we prove the correctness of lambda-lifting
(Theorem~\ref{thm:lambda-lifting-correctness},
p.~\pageref{thm:lambda-lifting-correctness}) by induction on the height of the
optimised reduction.

Section~\ref{sec:strong-invariants} defines stronger invariants and
rewords the correctness theorem with them.  Section~\ref{sec:overview}
gives an overview of the proof.  Sections~\ref{sec:rewriting-lemmas}
and~\ref{sec:aliasing-lemmas} prove a few lemmas needed for the proof.
Section~\ref{sec:proof-correctness} contains the actual proof of
correctness.

\subsubsection{Strengthened hypotheses}
\label{sec:strong-invariants}

We need strong induction hypotheses to ensure that key invariants about
stores and environments hold at every step.  For that purpose, we define
\emph{aliasing-free environments}, in which locations may not be
referenced by more than one variable, and \emph{local positions}.  They
yield a strengthened version of liftable parameters
(Definition~\ref{dfn:var-liftable}).  We then define lifted environments
(Definition~\ref{dfn:lifted-env}) to mirror the effect of
lambda-lifting in lifted terms captured in closures, and finally
reformulate the correctness of lambda-lifting in
Theorem~\ref{thm:correction-ll} with hypotheses strong enough to be
provable directly by induction.

\begin{definition}[Aliasing]\label{dfn:aliasing}
A set of environments $\mathcal{E}$ is \emph{aliasing-free}
when:
\[\forall \rho,\rho' \in \mathcal{E}, \forall x \in \dom(\rho), \forall y
\in \dom(\rho'),\
\rho\ x = \rho'\ y \Rightarrow x = y.
\]
By extension, an environment of functions $\F$ is aliasing-free when
$\Env(\F)$ is aliasing-free.  \qed
\end{definition}
The notion of aliasing-free environments is not an artifact of our small
language, but translates a fundamental property of the C semantics:
distinct function parameters or local variables are always bound to
distinct memory locations (Section~6.2.2, paragraph~6 in ISO/IEC 9899
\cite{iso9899}).  

A local position is any position in a term except inner functions.
Local positions are used to distinguish functions defined directly in a
term from deeper nested functions, because we need to enforce
Invariant~\ref{case:loc} (Definition~\ref{dfn:var-liftable}) on the
former only.
\begin{definition}[Local position]
\emph{Local positions} are defined inductively as follows:
\begin{enumerate}
\item $M$ is in local position in $M$, $x \coloneqq M$, $M \ ;\ M$,
  \ite{M}{M}{M} and $f(M,\dotsc,M)$.
\item $N$ is in local position in \letrec{f(\range{x}{1}{n})}{M}{N}.
\end{enumerate}
    \qed
\end{definition}

We extend the notion of liftable parameter
(Definition~\ref{dfn:var-liftable-simple},
p.~\pageref{dfn:var-liftable-simple}) to enforce invariants on stores and
environments.
\begin{definition}[Extended liftability]\label{dfn:var-liftable}
  The parameter $x$ is \emph{liftable} in $(M,\F,\rhot,\rho)$ when:
  \begin{enumerate}
    \item $x$ is defined as the parameter of a function $g$,
    either in $M$ or in $\F$,
    \label{case:def}
    \item in both $M$ and $\F$,
    inner functions in $g$, named $h_i$, are defined and called
    exclusively:
      \begin{enumerate}
      \item in tail position in $g$, or
      \item in tail position in some $h_j$ (with possibly $i=j$), or
      \item in tail position in $M$,
      \end{enumerate}
    \label{case:pos}
    \item for all $f$ defined in local position in $M$,
    $x \in \dom(\rhot\concat\rho) \Leftrightarrow \exists i, f = h_i$,
    \label{case:loc}
    \item moreover,
    if $h_i$ is called in tail position in $M$,
    then $x \in \dom(\rhot)$,
    \label{case:term}
    \item in \F,
    $x$ appears necessarily and exclusively in the environments of the
    $h_i$'s closures,
    \label{case:exclu}
    \item $\F$ contains only compact closures and
    $\Env(\F)\cup\{\rho,\rhot\}$ is aliasing-free.
    \label{case:share}
  \end{enumerate}
    \qed
\end{definition}

We also extend the definition of lambda-lifting
(Definition~\ref{dfn:lifted-term}, p.~\pageref{dfn:lifted-term}) to
environments, in order to reflect changes in lambda-lifted parameters
captured in closures.
\begin{definition}[Lifted form of an environment]\label{dfn:lifted-env}
  \begin{align*}
  \text{If } \F\,f =&
  \fun{\range{x}{1}{n}}{b}{\rho',\mathcal{F'}}\qquad\text{then}\\
  \lift{\F}\ f=&
  \begin{cases}
  \fun{\range{x}{1}{n}x}{\lift{b}}{\rho'|_{\dom(\rho')\setminus\{x\}},\lift{\mathcal{F'}}}&\text{when
  $f = h_i$ for some $i$}\\
  \fun{\range{x}{1}{n}}{\lift{b}}{\rho',\lift{\mathcal{F'}}}&\text{otherwise}
  \end{cases}
  \end{align*}
  \qed
\end{definition}
Lifted environments are defined such that a liftable parameter never
appears in them.  This property will be useful during the proof of
correctness.
\begin{lemma}\label{lem:Fstarclean}
  If $x$ is a liftable parameter in $(M,\F,\rhot,\rho)$,
  then $x$ does not appear in \lift{\F}.
\end{lemma}
\begin{proof}
  Since $x$ is liftable in $(M, \F, \rhot, \rho)$,
  it appears exclusively in the environments of $h_i$.
  By definition, it is removed when building \lift{\F}. \qed
\end{proof}

These invariants and definitions lead to an enhanced correctness theorem.
\begin{theorem}[Correctness of lambda-lifting]\label{thm:correction-ll}
If $x$ is a liftable parameter in $(M,\F,\rhot,\rho)$, then
\[\reduction{M}{s}{v}{s'} \text{ implies }
\reductionF{\lift{M}}{s}{v}{s'}{\lift{\F}}\]
\end{theorem}
Since naive and optimised reductions rules are equivalent
(Theorem~\ref{thm:sem-equiv}, p.~\pageref{thm:sem-equiv}), the proof of
Theorem~\ref{thm:lambda-lifting-correctness}
(p.~\pageref{thm:lambda-lifting-correctness}) is a direct corollary of this
theorem.
\begin{corollary}
If $x$ is a liftable parameter in $M$, then
\[\exists t,
\reductionNF[\varepsilon]{M}{\varepsilon}{v}{t}{\varepsilon} \text{ implies }
\exists t',
\reductionNF[\varepsilon]{\lift{M}}{\varepsilon}{v}{t'}{\varepsilon}.\]
\end{corollary}

\subsubsection{Overview of the proof}
\label{sec:overview}

With the enhanced liftability definition, we have strong enough
invariants to perform a proof by induction of the correctness theorem.
This proof is detailed in Section~\ref{sec:proof-correctness}.

The proof is not by structural induction but by induction on the height
of the derivation.  This is necessary because, even with the stronger
invariants, we cannot apply the induction hypotheses directly to the
premises in the case of the (call) rule: we have to change the stores
and environments, which means rewriting the whole derivation tree,
before using the induction hypotheses.

For that reason, the most important and difficult case of the proof is
the (call) rule.  We split it into two cases: calling one of the lifted
functions ($f = h_i$) and calling another function (either $g$, where
$x$ is defined, or any other function outside of $g$).  Only the former
requires rewriting; the latter follows directly from the induction
hypotheses.

In the (call) rule with $f = h_i$, issues arise when reducing the body
$b$ of the lifted function.  During this reduction, indeed, the store
contains a new location $l'$ bound by the environment to the lifted
variable $x$, but also contains the location $l$ which
contains the original value of $x$.  Our goal is to show that the
reduction of $b$ implies the reduction of $\lift{b}$, with store and
environments fulfilling the constraints of the (call) rule.

To obtain the reduction of the lifted body $\lift{b}$, we modify the
reduction of $b$ in a series of steps, using several lemmas:
\begin{itemize}
    \item the location $l$ of the free variable $x$ is moved to the tail
        environment (Lemma~\ref{lem:switch-x});
    \item the resulting reduction meets the induction hypotheses, which
        we apply to obtain the reduction of the lifted body $\lift{b}$;
    \item however, this reduction does not meet the constraints of the
        optimised reduction rules because the location $l$ is not fresh:
        we rename it to a fresh location $l'$ to hold the lifted
        variable;
    \item finally, since we renamed $l$ to $l'$, we need to reintroduce
        a location $l$ to hold the original value of $x$
        (Lemmas~\ref{lem:intro-in-store} and~\ref{lem:intro-in-env}).
\end{itemize}
The rewriting lemmas used in the (call) case are shown in
Section~\ref{sec:rewriting-lemmas}.

For every other case, the proof consists in checking thoroughly that the
induction hypotheses apply, in particular that $x$ is liftable in the premises.
It consists in checking invariants of the extended liftability definition
(Definition~\ref{dfn:var-liftable}).  To keep the main proof as compact as
possible, the most difficult cases of liftability, related to aliasing, are
proven in some preliminary lemmas (Section~\ref{sec:aliasing-lemmas}).

One last issue arises during the induction when one of the premises does
not contain the lifted variable $x$.  In that case, the invariants do
not hold, since they assume the presence of $x$.  But it turns out that
in this very case, the lifting function is the identity (since there is
no variable to lift) and lambda-lifting is trivially correct.

\subsubsection{Rewriting lemmas} \label{sec:rewriting-lemmas}
Calling a lifted function has an impact on the resulting store: new
locations are introduced for the lifted parameters and the earlier
locations, which are not modified anymore, are hidden.  Because of these
changes, the induction hypotheses do not apply directly in the case of
the (call) rule for a lifted function $h_i$. We use the following three
lemmas to obtain, through several rewriting steps, a reduction of lifted
terms meeting the induction hypotheses.

\begin{itemize}
    \item Lemma~\ref{lem:switch-x} shows that moving a variable from the non-tail
environment $\rho$ to the tail environment $\rhot$ does not change the
result, but restricts the domain of the store.  It is used transform the
original free variable $x$ (in the non-tail environment) to its lifted
copy (which is a parameter of $h_i$, hence in the tail environment).
\item Lemmas~\ref{lem:intro-in-store} and~\ref{lem:intro-in-env}
    add into the store and the environment a fresh location, bound to an
    arbitrary value.  It is used to reintroduce the location containing
    the original value of $x$, after it has been alpha-converted to
    $l'$.
\end{itemize}

\begin{lemma}[Switching to tail environment]\label{lem:switch-x}
  If $\reduction[\env{\rhot}{(x,l)\concat\rho}]{M}{s}{v}{s'}$ and $x \notin
  \dom(\rhot)$ then
  $\reduction[\env{\rhot\concat(x,l)}{\rho}]{M}{s}{v}{s'|_{\dom(s')\setminus\{l\}}}$.
  Moreover, both derivations have the same height.
\end{lemma}
\begin{proof}
  By induction on the structure of the derivation.
  For the (val), (var), (assign) and (call) cases, we use the fact that
  $\gc{\rhot\concat(x,l)}{s} = s'|_{\dom(s')\setminus\{l\}}$ when
  $s' = \gc{\rhot}{s}$.
  \qed
\end{proof}
\begin{lemma}[Spurious location in store]\label{lem:intro-in-store}
  If $\reduction{M}{s}{v}{s'}$ and $k$ does not appear in either $s$, $\F$
  or $\rhot\concat\rho$,
  then, for all value $u$,
  $\reduction{M}{\ajout{s}{k}{u}}{v}{\ajout{s'}{k}{u}}$.
  Moreover, both derivations have the same height.
\end{lemma}
\begin{proof}
  By induction on the height of the derivation.  
  The key idea is to add $(k,u)$ to every store in the derivation tree.
  A collision might occur in the (call) rule, if there is some  $j$ such
  that $l_j = k$.  In that case, we need to rename $l_j$ to some fresh
  variable $l'_j \neq k$ (by alpha-conversion) before applying the
  induction hypotheses.
      \qed
\end{proof}
\begin{lemma}[Spurious variable in environments]\label{lem:intro-in-env}
  \begin{align*}
  \forall l,l', \reduction[\env{\rhot\concat(x,l)}{\rho}]{M}{s}{v}{s'} \ssi&
  \reduction[\env{\rhot\concat(x,l)}{(x,l')\concat\rho}]{M}{s}{v}{s'}
  \end{align*}
  Moreover, both derivations have the same height.
\end{lemma}
\begin{proof}
 By induction on the structure of the derivation.  The proof relies solely
 on the fact that $\rhot\concat(x,l)\concat\rho =
 \rhot\concat(x,l)\concat(x,l')\concat\rho$.
       \qed
\end{proof}

\subsubsection{Aliasing lemmas} \label{sec:aliasing-lemmas}
We need two lemmas to show that environments remain aliasing-free during
the proof by induction in Section~\ref{sec:proof-correctness}.
They are purely technical lemmas that consist in proving that the aliasing
invariant (Invariant~\ref{case:share}, Definition~\ref{dfn:var-liftable})
holds in the context of the (call) and (letrec) rules, respectively.
We only show the (call) lemma here; the (letrec) lemma is very similar and
detailed in the technical report \cite{kerneis12}.

\begin{lemma}[Aliasing in (call) rule]\label{lem:aliasing-call}
Assume that, in a (call) rule,
\begin{itemize}
    \item $\F\,f = \fun{\range{x}{1}{n}}{b}{\rho',\mathcal{F'}}$,
    \item $\Env(\F)$ is aliasing-free, and
    \item $\rho''= \drange{x}{l}{1}{n}$, with fresh and
    distinct locations $l_{i}$.
\end{itemize}
Then
$\Env(\ajout{\mathcal{F'}}{f}{\F\,f})\cup\{\rho',\rho''\}$ is also aliasing-free.
\end{lemma}
\begin{proof}
  Let $\mathcal{E} = \Env(\ajout{\mathcal{F'}}{f}{\F\,f})\cup\{\rho'\}$.
  We know that $\mathcal{E}\subset\Env(\F)$ so $\mathcal{E}$ is aliasing-free
  We want to show that adding fresh and distinct locations from
  $\rho''$ preserves this lack of aliasing.  More precisely,
  we want to show that
  \begin{align*}
  \forall \rho_1,\rho_2 \in \mathcal{E}\cup\{\rho''\},
  \forall x \in \dom(\rho_1),
  \forall y \in \dom(\rho_2),\
  \rho_1\ x = \rho_2\ y \Rightarrow x = y\\
  \intertext{given that}
  \forall \rho_1,\rho_2 \in \mathcal{E},
  \forall x \in \dom(\rho_1),
  \forall y \in \dom(\rho_2),\
  \rho_1\ x = \rho_2\ y \Rightarrow x = y.
  \end{align*}
  We reason by checking of all cases.
  If $\rho_1\in \mathcal{E}$ and $\rho_2\in \mathcal{E}$, immediate.
  If $\rho_1 = \rho_2 = \rho''$ then
  $\rho''\ x = \rho''\ y \Rightarrow x = y$
  holds because the locations of $\rho''$ are distinct.
  If $\rho_1 = \rho''$ and
  $\rho_2 \in \mathcal{E}$ then
  $\rho_1\ x = \rho_2\ y \Rightarrow x = y$
  holds because
  $\rho_1\ x \neq \rho_2\ y$
  (by freshness hypothesis). \qed
\end{proof}

\subsubsection{Proof of correctness} \label{sec:proof-correctness}

We finally recall and show Theorem~\ref{thm:correction-ll}
(p.~\pageref{thm:correction-ll}).
\setcounter{theorem}{2}
\begin{theorem}[Correctness of lambda-lifting]
If $x$ is a liftable parameter in $(M,\F,\rhot,\rho)$, then
\[\reduction{M}{s}{v}{s'} \text{ implies }
\reductionF{\lift{M}}{s}{v}{s'}{\lift{\F}}\]
\end{theorem}

Assume that $x$ is a liftable parameter in $(M,\F,\rhot,\rho)$. The proof is by
induction on the height of the reduction of $\reduction{M}{s}{v}{s'}$.  We only
show the case (call).  The full proof is available in the technical report
\cite{kerneis12}.

\paragraph{(call) --- first case}
    First, we consider the most interesting case where there exists $i$ such
    that $f = h_i$.
    The variable
    $x$ is a liftable parameter in $(h_i(\range{a}{1}{n}),\F,\rhot,\rho)$
    hence in $(a_i,\F,\varepsilon,\rhot\concat\rho)$ too.

    By the induction hypotheses, we get
    \[\reduclift[\env{}{\rhot\concat\rho}]{a_{i}}{s_{i}}{v_{i}}{s_{i+1}}.\]
    By the definition of lifting, $\lift{h_i(\range{a}{1}{n})} =
    h_i(\lift{a_{1}},\dotsc,\lift{a_{n}},x)$.  But $x$ is not a liftable
    parameter in $(b,\mathcal{F'},\rho'',\rho')$ since the
    Invariant~\ref{case:term} might be broken: $x \notin \dom(\rho'')$
    ($x$ is not a parameter of $h_i$) but $h_j$ might appear in tail
    position in $b$.

    On the other hand, we have $x \in \dom(\rho')$: since, by
    hypothesis, $x$ is a liftable parameter in
    $(h_i(\range{a}{1}{n}),\F,\rhot,\rho)$, it appears necessarily  in
    the environments of the closures of the $h_i$, such as $\rho'$.
    This allows us to split $\rho'$ into two parts:
    $\rho' = (x,l)\concat\rho'''$.
    It is then possible to move $(x,l)$ to the tail environment,
    according to Lemma~\ref{lem:switch-x}:
    \[\reductionF[\env{\rho''(x,l)}{\rho'''}]{b}{\ajout{s_{n+1}}{l_i}{v_i}}{v}{s'|_{\dom(s')\setminus\{l\}}}{\
    \ajout{\mathcal{F'}}{f}{\F\,f}}\]
    This rewriting ensures that
    $x$ is a liftable parameter in
    $(b,\ajout{\mathcal{F'}}{f}{\F\,f},\rho''\concat(x,l),\rho''')$ (by
    Lemma~\ref{lem:aliasing-call} for the Invariant~\ref{case:share}).

    By the induction hypotheses,
    \[\reductionF[\env{\rho''(x,l)}{\rho'''}]{\lift{b}}{\ajout{s_{n+1}}{l_i}{v_i}}{v}{s'|_{\dom(s')\setminus\{l\}}}{\
    \lift{\ajout{\mathcal{F'}}{f}{\F\,f}}}\]
    The $l$ location is not fresh: it must be rewritten into a fresh
    location, since $x$ is now a parameter of $h_i$.
    Let $l'$ be a location appearing in neither 
    $\lift{\ajout{\mathcal{F'}}{f}{\F\,f}}$, nor $\ajout{s_{n+1}}{l_i}{v_i}$ or $\rho''\concat\rhot'$.
    Then $l'$ is a fresh location, which is to act as $l$ in the
    reduction of $\lift{b}$.

    We will show that, after the reduction, $l'$ is not in the
    store (just like $l$ before the lambda-lifting).  In the meantime, the value
    associated to $l$ does not change (since $l'$ is modified instead
    of $l$).

    Lemma~\ref{lem:Fstarclean} implies that $x$  does not appear in the
    environments of \lift{\F}, so it does not appear in the environments
    of $\lift{\ajout{\mathcal{F'}}{f}{\F\,f}}\subset\lift{\F}$ either.
    As a consequence, lack of aliasing implies by Definition~\ref{dfn:aliasing}
    that the label $l$, associated to $x$, does not appear in
    $\lift{\ajout{\mathcal{F'}}{f}{\F\,f}}$ either, so
    \[\lift{\ajout{\mathcal{F'}}{f}{\F\,f}}[l'/l] = \lift{\ajout{\mathcal{F'}}{f}{\F\,f}}.\]
    Moreover, $l$ does not appear in $s'|_{\dom(s')\setminus\{l\}}$.
    Since $l'$ does not appear in the store or the
    environments of the reduction, we rename $l$ to $l'$:
    \[\reductionF[\env{\rho''(x,l')}{\rho'''}]{\lift{b}}{\ajout{s_{n+1}[l'/l]}{l_i}{v_i}}{v}{s'|_{\dom(s')\setminus\{l\}}}{\
    \lift{\ajout{\mathcal{F'}}{f}{\F\,f}}}.\]
    We want now to reintroduce $l$.  
    Let $v_x = s_{n+1}\ l$.  The location $l$  does not appear in
    $\ajout{s_{n+1}[l'/l]}{l_i}{v_i}$, $\lift{\ajout{\mathcal{F'}}{f}{\F\,f}}$, or $\rho''(x,l')\concat\rho'''$.
    Thus, by Lemma~\ref{lem:intro-in-store},
    \[\reductionF[\env{\rho''(x,l')}{\rho'''}]{\lift{b}}{\ajout{\ajout{s_{n+1}[l'/l]}{l_i}{v_i}}{l}{v_x}}{v}{\ajout{s'|_{\dom(s')\setminus\{l\}}}{l}{v_x}}{\
    \lift{\ajout{\mathcal{F'}}{f}{\F\,f}}}.\]
    Since
    \begin{align*}
      \ajout{\ajout{s_{n+1}[l'/l]}{l_i}{v_i}}{l}{v_x}
      &= \ajout{\ajout{s_{n+1}[l'/l]}{l}{v_x}}{l_i}{v_i}
      && \text{because $\forall i, l \neq l_i$}\\
      &= \ajout{\ajout{s_{n+1}}{l'}{v_x}}{l_i}{v_i}
      && \text{because $v_x = s_{n+1} l$}\\
      &= \ajout{\ajout{s_{n+1}}{l_i}{v_i}}{l'}{v_x}
      && \text{because $\forall i, l' \neq l_i$}
    \end{align*}
    and
    $\ajout{s'|_{\dom(s')\setminus\{l\}}}{l}{v_x} = \ajout{s'}{l}{v_x}$,
    we finish the rewriting by Lemma~\ref{lem:intro-in-env},
    \[\reductionF[\env{\rho''(x,l')}{(x,l)\concat\rho'''}]{\lift{b}}{\ajout{\ajout{s_{n+1}}{l_i}{v_i}}{l'}{v_x}}{v}{\
    \ajout{s'}{l}{v_x}}{\lift{\ajout{\mathcal{F'}}{f}{\F\,f}}}.\]
    Hence the result:
    \[\inferrule*[Left=(call)]{\
    \lift{\F}\ h_i = \fun{\range{x}{1}{n}x}{\lift{b}}{\rho',\lift{\mathcal{F'}}}\\
    \rho''= \drange{x}{l}{1}{n}(x,\rhot\ x)\\
    \text{$l'$ and $l_{i}$ fresh and distinct}\\\\
    \forall i,\reduclift[\env{}{\rhot\concat\rho}]{a_{i}}{s_{i}}{v_{i}}{s_{i+1}} \\
    \reduclift[\env{}{\rhot\concat\rho}]{x}{s_{n+1}}{v_x}{s_{n+1}} \\
    \reductionF[\env{\rho''(x,l')}{\rho'}]{\lift{b}}{\ajout{\ajout{s_{n+1}}{l_i}{v_i}}{l'}{v_x}}{v}{\
    \ajout{s'}{l}{v_x}}{\lift{\ajout{\mathcal{F'}}{f}{\F\,f}}}
    }{\
    \reduclift{h_i(\range{a}{1}{n})}{s_{1}}{v}{\gc{\rhot}{\ajout{s'}{l}{v_x}}}}\]
    Since $l \in \dom(\rhot)$ (because $x$ is a liftable parameter in $(h_i(\range{a}{1}{n}),\F,\rhot,\rho)$),
    the extraneous location is reclaimed as expected:
    $\gc{\rhot}{\ajout{s'}{l}{v_x}} = \gc{\rhot}{s'}$.

\paragraph{(call) --- second case}
    We now consider the case where $f$ is not one of the $h_i$.
    The variable
    $x$ is a liftable parameter in $(f(\range{a}{1}{n}),\F,\rhot,\rho)$
    hence in
    $(a_i,\F,\varepsilon,\rhot\concat\rho)$ too.

    By the induction hypotheses, we get
    \[\reduclift[\env{}{\rhot\concat\rho}]{a_{i}}{s_{i}}{v_{i}}{s_{i+1}},\]
    and, by Definition~\ref{dfn:lifted-term},
    \[\lift{f(\range{a}{1}{n})} = f(\lift{a_{1}},\dotsc,\lift{a_{n}}).\]
    If $x$ is not defined in $b$ or $\F$, then $\lift{}$ is the identity
    function and can trivially be applied to the reduction of $b$.  Otherwise,
    $x$ is a liftable parameter in
    $(b,\ajout{\mathcal{F'}}{f}{\F\,f},\rho'',\rho')$ --- when checking the
    invariants of Definition~\ref{dfn:var-liftable}, we use
    Lemma~\ref{lem:aliasing-call} for the Invariant~\ref{case:share} and
    check separately $f = g$ and $f \neq g$ for the Invariants~\ref{case:loc}
    and~\ref{case:term} (see the technical report \cite{kerneis12} for more
    details).

    By the induction hypotheses,
    \[\reductionF[\env{\rho''}{\rho'}]{\lift{b}}{\ajout{s_{n+1}}{l_{i}}{v_{i}}}{v}{s'}{\lift{\ajout{\mathcal{F'}}{f}{\F\,f}}}\]
    hence:
    \[\inferrule*[Left=(call)]{\
    \lift{\F}\ f = \fun{\range{x}{1}{n}}{\lift{b}}{\rho',\lift{\mathcal{F'}}}\\
    \rho''= \drange{x}{l}{1}{n}\\
    \text{$l_{i}$ fresh and distinct}\\\\
    \forall i,\reduclift[\env{}{\rhot\concat\rho}]{a_{i}}{s_{i}}{v_{i}}{s_{i+1}} \\
    \reductionF[\env{\rho''}{\rho'}]{\lift{b}}{\ajout{s_{n+1}}{l_{i}}{v_{i}}}{v}{s'}{\lift{\ajout{\mathcal{F'}}{f}{\F\,f}}}
    }{\
    \reduclift{f(\range{a}{1}{n})}{s_{1}}{v}{\gc{\rhot}{s'}}}\]

\paragraph{Other cases}
The detailed proof of the other cases is available in the technical report
\cite{kerneis12}.  They are mostly straightforward by induction, except (letrec)
which requires an additional aliasing lemma to check the liftability invariants.
\qed

\section{Conclusions and further work}
\label{sec:conclusion}

In this paper, we have described CPC, a programming language that
provides threads which are implemented, in different parts of the
program, either as extremely lightweight heap-allocated data structures,
or as native operating system threads.  The compilation technique used
by CPC is somewhat unusual, since it involves a continuation-passing
style (CPS) transform for the C programming language.  We have shown the
correctness of that particular instance of the CPS transform, as well as
the correctness of CPC's compilation scheme; while other efficient
systems for programming concurrent systems exist, we claim that the
existence of a proof of correctness makes CPC unique among them.

CPC is highly adapted to writing high-performance network servers.  To
prove this fact, we have written Hekate, a large scale BitTorrent seeder
in CPC.  Hekate has turned out to be a compact, maintainable, fast and
reliable piece of software.

We enjoyed writing Hekate very much.  Due to the lightweight threads that
it provides, and due to the determinacy of scheduling of attached threads,
CPC threads have a very different feel from threads in other programming
languages; discovering the right idioms and the right abstractions for CPC
has been (and remains) one of the most enjoyable parts of our work.

For CPC to become a useful production language it must come equipped
with a consistent and powerful standard library that encapsulates useful
programming idioms in a generally usable form.  The current CPC library
has been written on an on-demand basis, mainly to meet the needs of
Hekate; the choice of functions that it provides is therefore somewhat
random.  Filling in the holes of the library should be a fairly
straightforward job.  For CPC to scale easily to multiple cores, this
extended standard library should also offer the ability to run several
event loops, scheduled on different cores, and migrate threads between
them.

We have no doubt that CPC can be useful for applications other than
high-performance network servers.  One could for example envision a GUI
system where every button is implemented using three CPC threads: one that
waits for mouse clicks, one that draws the button, and one that coordinates
with the rest of the system.  To be useful in practice, such a system
should be implemented using a standard widget library; the fact that CPC
integrates well with external event loops indicates that this should be
possible.

Finally, the ideas used in CPC might be applicable to programming languages
other than C.  For example, a continuation passing transform might
be a way of extending Javascript with threads without the need to deploy
a new Javascript runtime to hundreds of millions of web browsers.

\section*{Software availability}

The full CPC translator, including sources and benchmarking code, is
available online at
\url{http://www.pps.univ-paris-diderot.fr/~kerneis/software/cpc/}.\\
The sources of Hekate are available online at
\url{http://www.pps.univ-paris-diderot.fr/~kerneis/software/hekate/}.

\begin{acknowledgements}
The authors would like to thank the anonymous reviewers for their valuable
comments and suggestions to improve this article.
\end{acknowledgements}

\end{document}